\documentclass[pra,aps,longbibliography,twocolumn,superscriptaddress,floatfix]{revtex4-1}
\usepackage{amsmath,amssymb,graphicx,units}
\usepackage[usenames,dvipsnames]{color}
\usepackage{verbatim}
\usepackage{lipsum}
\usepackage{setspace}
\usepackage{dcolumn}
\usepackage{bm}
\usepackage[bookmarks=false,linkcolor=blue,urlcolor=blue,colorlinks,citecolor=blue]{hyperref}

\usepackage{mathrsfs}
\usepackage[T1]{fontenc}

\usepackage{amsthm}
\usepackage{graphicx}

\newtheorem{theorem}{Theorem}

\newtheorem{definition}{Definition}
\newtheorem{proposition}{Proposition}

\begin{document}

\title{Intrinsic preservation of plasticity in continual quantum  learning}

\author{Yu-Qin Chen}
\email{yqchen@gscaep.ac.cn}
\affiliation{Graduate School of China Academy of Engineering Physics, Beijing 100193, China}

\author{Shi-Xin Zhang}
\email{shixinzhang@iphy.ac.cn}
\affiliation{Institute of Physics, Chinese Academy of Sciences, Beijing 100190, China}

\date{\today}

\begin{abstract}

Artificial intelligence in dynamic, real-world environments requires the capacity for continual learning. However, standard deep learning suffers from a fundamental issue: loss of plasticity, in which networks gradually lose their ability to learn from new data. Here we show that quantum learning models naturally overcome this limitation, preserving plasticity over long timescales. We demonstrate this advantage systematically across a broad spectrum of tasks from multiple learning paradigms, including supervised learning and reinforcement learning, and diverse data modalities, from classical high-dimensional images to quantum-native datasets. Although classical models exhibit performance degradation correlated with unbounded weight and gradient growth, quantum neural networks maintain consistent learning capabilities regardless of the data or task. We identify the origin of the advantage as the intrinsic physical constraints of quantum models. Unlike classical networks where unbounded weight growth leads to landscape ruggedness or saturation, the unitary constraints confine the optimization to a compact manifold. Our results suggest that the utility of quantum computing in machine learning extends beyond potential speedups, offering a robust pathway for building adaptive artificial intelligence and lifelong learners.
\end{abstract}

\maketitle


\section{Introduction}

Recent advances in artificial intelligence (AI)~\cite{lecun2015deep, jordan2015machine}, particularly in large language models, rely on scaling static architectures on massive fixed datasets \cite{Bommasani2021, OpenAI2023, Guo2025ds}. However, the deployment of these systems in real world reveals a critical limitation: the world is non-stationary, and data distributions evolve continuously over time. It is computationally prohibitive to retrain large-scale models from scratch to accommodate every distributional shift and new knowledge~\cite{Strubell2020}. Consequently, the paradigm of continual learning (CL)  emerges~\cite{Ditzler2015, Hadsell2020, Wang2024clreview}, where models learn sequentially from a stream of tasks. The central challenge of CL lies in the trade-off between \textit{stability} and \textit{plasticity}. Stability requires the model to keep proficiency on previously learned tasks, preventing catastrophic forgetting \cite{Goodfellow2015, Kirkpatrick2017, Zenke2017}, whereas plasticity demands the flexibility to efficiently learn new knowledge from fresh data streams. While the community has dedicated significant effort to the former challenge of stability, a more subtle problem has recently been identified regarding the latter, dubbed as loss of plasticity \cite{Dohare2024}. Recent studies demonstrate that standard deep learning methods not only forget; they also progressively lose their learnability, with the performance on subsequent tasks decaying continuously~\cite{Dohare2024, Chaudhry2018, Ash2020warm, Berariu2021, Nikishin2022rl, Abbas2023rlplasticity, Lyle2023platsicity, Elsayed2024}.

Recently, quantum machine learning (QML) has emerged as a promising alternative learning paradigm, leveraging quantum principles to process information in high-dimensional Hilbert spaces~\cite{Biamonte2017a}. Theoretical foundations have established the expressive power and generalization capabilities of quantum models~\cite{Lloyd2013, Schuld2014quest, Huang2020powerdata, Huang2022science, Cerezo2022ncs}. Specifically, variational quantum circuits (VQCs) and quantum neural networks (QNNs) have been successfully applied to diverse tasks, including classification~\cite{Havlicek2019nature, Prez-Salinas2020reupload}, generative modeling~\cite{Liu2018born, Dallaire-Demers2018qgan, Zhang2024diffusion}, and quantum-native data processing~\cite{Zhang2021vqnhe, Cong2019}. Furthermore, recent studies explore sophisticated architectures, optimization strategies, and utilize the intrinsic information properties to enhance QML performance and resilience~\cite{Li2017control, Beer2020dnn, Shen2020infor, Banchi2021qi, Li2022recent, Zheng2023fqnn, Jger2023nc, Miao2023nnvqa, Wu2024prl, Li2025vqnhe, Chen2025resilience, Cui2025}. In the context of continual learning, initial investigations focused primarily on the stability aspect of CL. Previous studies demonstrated that QML approaches based on variational quantum circuits (VQCs) cannot naturally mitigate catastrophic forgetting, exhibiting forgetting dynamics similar to classical counterparts under standard training settings~\cite{Jiang2022, Zhang2024clq}. However, the plasticity nature of quantum models remains unexplored. Therefore, it is an urgent and relevant question whether the unique structure of quantum neural networks can offer a decisive plasticity advantage in maintaining learnability over time. 
A positive answer would suggest that quantum computing offers an interesting mechanism toward adaptive intelligence, enabling systems that can adapt to dynamic environments without the huge cost of retraining or the risk of classical plasticity loss.

\begin{figure*}[t]\centering
	\includegraphics[width=\textwidth]{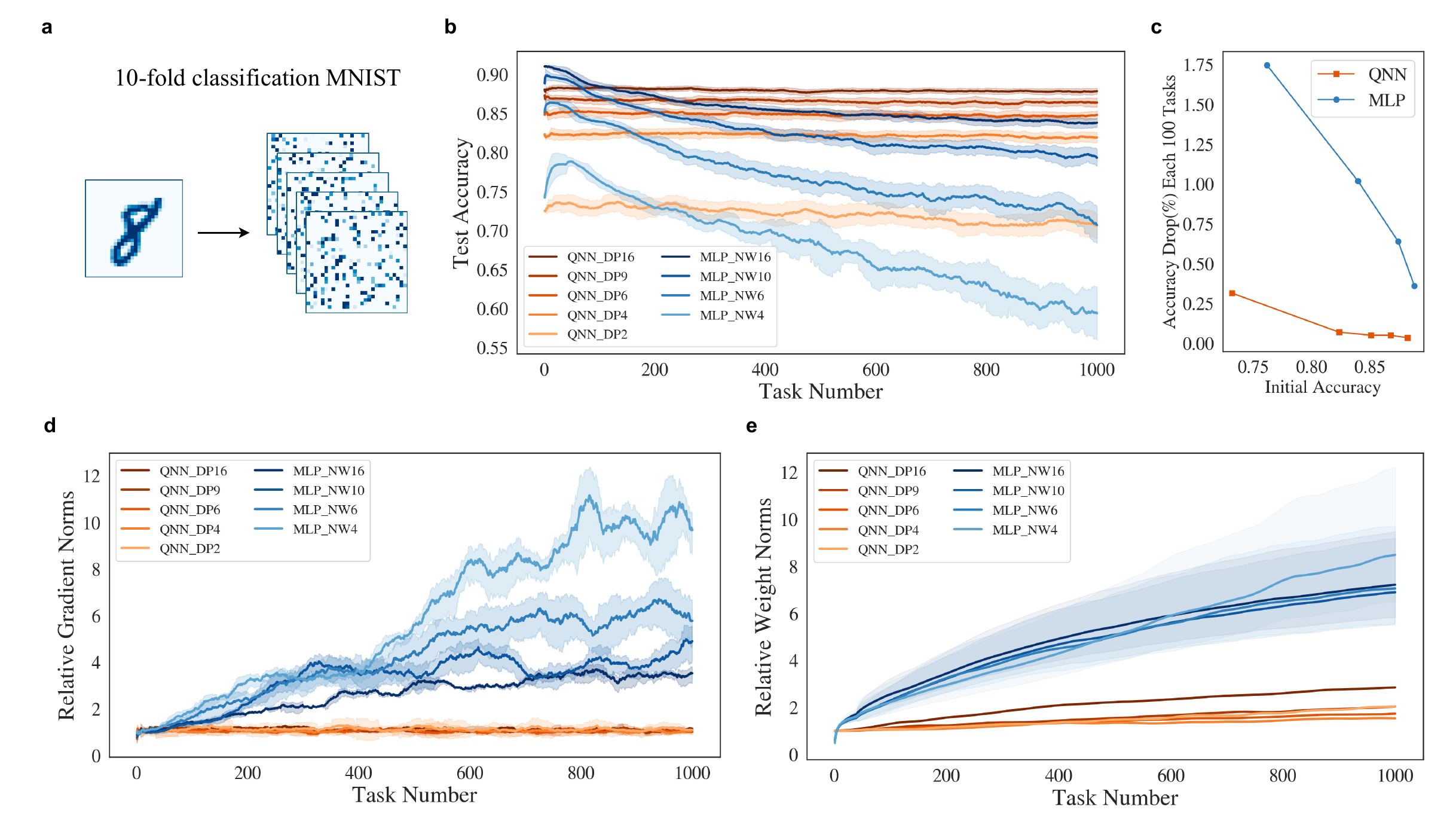}
	\caption{\textbf{Quantum neural networks maintain plasticity while classical neural networks fail in deep continual learning on permuted MNIST.} 
\textbf{a}, Illustration of the permuted MNIST task for continual learning. A sequence of 10-fold classification tasks is generated by applying different random permutations to the pixels of the MNIST digits, forcing the model to continuously adapt to new input distributions. 
\textbf{b}, Test accuracy as a function of the task number for QNNs of varying circuit depths (DP) and MLPs of varying network widths (NW). QNNs maintain a stable test accuracy across the whole task sequence, demonstrating robust plasticity. In contrast, all MLP variants exhibit a significant and continuous degradation in performance, indicating a severe loss of plasticity. 
\textbf{c}, Average accuracy drop against the model's initial accuracy averaged over the first 200 tasks. MLPs show a much higher rate of plasticity loss, and it is more severe for models with worse initial performance. 
 \textbf{d}, Relative gradient norms and \textbf{e}, relative weight norms during the continual learning process. The norms are normalized by the average norm across the first 10 tasks. For MLPs, both gradient and weight norms grow unboundedly over time, correlating with their performance decline. For QNNs, these norms remain stable and bounded, providing a mechanism explanation for the preserved plasticity. 
 All curves show the moving average over 40-task windows, with shaded areas representing the standard deviation of the underlying data within each window.}
\label{fig:mnist}
\end{figure*}

In this work, we address the critical open question through a systematic comparison of classical multi-layer perceptrons (MLPs) and deep quantum neural networks (QNNs), which are the representative models for classical and quantum AI, respectively. We evaluate their long-term adaptability across a variety of continual learning paradigms including supervised classification and reinforcement learning. The evaluation uses both industrial standard classical benchmarks and quantum-native datasets. Unlike the small-scale proof-of-concept demonstrations typical of prior QML research, our framework involves training very deep parameterized quantum circuits end-to-end across thousands of sequential tasks. This represents a computational workload orders of magnitude more challenging than standard benchmarks, validating QML plasticity in a more realistic and scalable regime. Our analysis reveals a qualitative difference in their learning behaviors. Classical models show loss of plasticity as expected with unbounded weight growth. On the contrary, as long as the quantum model is trainable, namely, capable of bypassing static trainability hurdles such as barren plateaus, its unitary character naturally preserves the model's plasticity, regardless of the training duration. We attribute this plasticity to the unitary nature of evolution that confines optimization to a compact parameter manifold. While the compactness of quantum parameter manifolds is a well-known physical property, establishing its nontrivial algorithmic consequence for mitigating plasticity loss constitutes a central conceptual contribution of this work. By preventing the saturation dynamics that paralyzes classical networks, our results establish a distinct structural advantage for QNNs. This shifts the focus from computational speed to inherent adaptability, identifying quantum machine learning as a promising ingredient for realizing truly adaptive, lifelong learning agents.


\begin{figure*}[t]\centering
	\includegraphics[width=\textwidth]{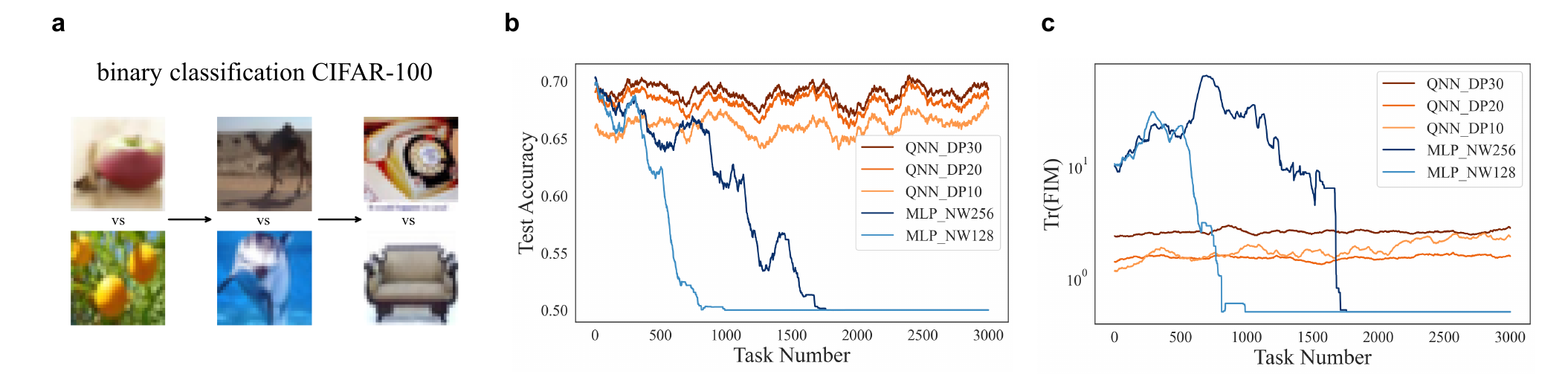}
	\caption{\textbf{Preserved plasticity of QNNs in continual learning of real-world image classification tasks.} 
\textbf{a}, Illustration of the sequential binary classification task constructed from the CIFAR-100 dataset. At each task, two classes are randomly sampled from the 100 available classes to form a new binary classification problem. This setting challenges the model's ability to continuously adapt to new visual concepts.
\textbf{b}, Test accuracy as a function of task number for QNNs of varying circuit depths (DP) and MLPs of varying network widths (NW). Similar to the results on permuted MNIST, QNNs maintain a stable test accuracy throughout the 3,000 tasks. In stark contrast, MLPs suffer a catastrophic loss of plasticity, with their performance rapidly decaying to guess level (50\% accuracy).
\textbf{c}, Mechanism insight by the trace of the Fisher Information Matrix, a measure of the model's effective learning capacity. $\text{Tr}(\text{FIM})$ for QNNs remains stable, indicating that their ability to learn is preserved. For MLPs, the $\text{Tr}(\text{FIM})$ collapses, quantitatively demonstrating that the model has lost its capacity to effectively update its parameters, which explains the performance decay in panel \textbf{b}. All curves show the moving average over 150-task windows.}
\label{fig:cifar}
\end{figure*}

\section{Preserved plasticity in supervised learning}

To systematically evaluate the plasticity of QNNs versus classical counterparts, we first deploy them on the standard continual learning benchmark of Permuted MNIST. We note that more advanced architectures, such as convolutional neural networks and ResNet \cite{He2015}, have already been shown to suffer from severe plasticity loss in this setting \cite{Dohare2024}. In our experimental setup, we generate a sequence of 1,000 tasks, where each task requires 10-fold classification on hand-written MNIST digits subject to a task-specific random pixel permutation (Fig.~\ref{fig:mnist}{\textbf{a}}). We compare a classical MLP with ReLU activation against a VQC constructed using $SU(4)$ entangling gates \cite{Dilip2022, Shen2024}. The MLP is configured with hidden widths (NW) ranging from 4 to 256 neurons, while the QNN utilizes varying circuit depths (DP) of 2 to 16 layers. For the QNN output, we employ a probability-based readout strategy, i.e., the logarithmic measurement probabilities of the first ten bitstrings are mapped to class logits for softmax function output, analogous to the classical output layer. Both models are trained using gradient-based optimization on each task's training set and evaluated on a held-out test set.

As shown in Fig.~\ref{fig:mnist}{\textbf{b}}, classical MLPs exhibit a continuous degradation in performance. Despite the sufficient learning capacity for each individual task, the test accuracy decays significantly as the task sequence progresses. On the contrary, QNNs maintain a stable accuracy throughout the entire CL sequence. Fig.~\ref{fig:mnist}{\textbf{c}} quantifies this difference: while MLPs suffer a large accuracy drop, QNNs exhibit negligible degradation. We verify that these divergent learning patterns persist across different model sizes and training configurations. 

To probe the underlying mechanism behind this divergence, we monitor the internal dynamics of the networks during CL. We compute the relative $L_2$ norms of the weight matrices ($||\theta||_2$) and the gradients ($||\nabla_\theta \mathcal{L}||_2$) at each task (Figs.~\ref{fig:mnist}{\textbf{d}}, \textbf{e}). In the classical case, we observe a monotonic, unbounded growth in weight magnitude. This weight explosion severely limits plasticity by pushing neurons into inactive regimes (e.g., dead ReLUs), leading to a corresponding collapse in learning capacity and effective rank. In direct contrast, the parameters of our $SU(4)$-based QNN are rotation angles on a compact Lie group manifold. This unitary constraint ensures that the parameter space remains bounded and periodic. Consequently, both the weight and gradient norms of the QNN are stable, preventing the saturation dynamics that paralyze classical ones.

To confirm that plasticity preservation is not unique to the $SU(4)$ quantum ansatz, we also test a restricted hardware-efficient ansatz \cite{Kandala2017} for QNN. We find identical plasticity preservation (see Appendix), confirming that the advantage is a property of the quantum model nature rather than a specific circuit architecture.

We extend our evaluation to a more challenging visual domain using the Split CIFAR-100 benchmark (Fig.~\ref{fig:cifar}\textbf{a}). We build a sequence of 3,000 binary classification tasks, where each task is to distinguish between two randomly sampled classes, such as apple versus orange. Inputs are encoded into quantum states via amplitude encoding after L2-normalization, and the output logit of the QNN is determined by the expectation value of a Pauli $Z$ operator on the last qubit. Notably, we utilize deep QNNs with up to 30 layers (more than 4,000 quantum trainable parameters) on this real-world dataset, demonstrating the scalability of the finding beyond typical toy-model experiments. 

Consistent with the MNIST results, standard MLPs suffer a catastrophic loss of plasticity, with accuracy collapsing toward random guessing after approximately 1,000 tasks (Fig.~\ref{fig:cifar}\textbf{b}). Crucially, for a fair comparison, both models are trained using exactly identical optimization hyperparameters and learning schedules, ensuring that the quantum advantage originates from the ansatz geometry rather than training hyperparameters. To further isolate the structural benefit of quantum models, we also evaluate an L2-regularized classical MLP with weight decay set to $10^{-4}$. Even with explicit regularization, the classical model still collapses to random guessing accuracy within the same timeframe. In stark contrast, QNNs demonstrate superior robustness, maintaining high classification accuracy throughout the extensive 3,000-task process. 

We also verify the robustness of this advantage on hardware-realistic noisy intermediate-scale quantum (NISQ) devices by simulating depolarizing noise ($p=5 \times 10^{-4}$). The noisy QNN maintains the learnability, confirming that the geometric preservation of plasticity is robust to small quantum errors.

We further train a classical MLP with a periodic sine activation function \cite{Sitzmann2020} on Split CIFAR-100. While this introduces bounded activations similar to quantum amplitudes, the weights themselves remain unbounded in Euclidean space. The Sin-MLP still suffers from severe plasticity loss and overfitting, as the unbounded weight growth drives the optimization landscape into high-frequency oscillation regimes (see Appendix). This confirms that true plasticity preservation requires the compact parameter manifold inherent to the quantum ansatz, not merely bounded output functions.

To provide a rigorous theoretical metric for learnability, we compute the trace of the Fisher information matrix (FIM) on a fixed probe dataset throughout the training (Fig.~\ref{fig:cifar}\textbf{c}) \cite{Achille2017}. While the FIM trace is an information-theoretic measure of a model's capacity, we specifically employ the empirical FIM in the main text, which represents the curvature of the loss surface with respect to the labels of the current batch. This metric serves as a direct proxy for the model's update agility. For MLPs, $\text{Tr}(\text{FIM})$ undergoes a rapid collapse, quantitatively demonstrating that the effective dimensionality of the optimization landscape shrinks over time. The classical network physically possesses parameters, but they become functionally dead. QNNs instead maintain a stable $\text{Tr}(\text{FIM})$, indicating that their effective learning capacity remains constant regardless of the time in the continual learning history. We provide a more fundamental analysis using the model FIM ($y$ from the model predicting instead of data ground truth in the empirical FIM case) and rank analysis in the Appendix, confirming that both metrics capture the same structural divergence in plasticity.

We further elevate these empirical findings to a rigorous theoretical analysis in Appendix. By analyzing the asymptotic geometry of the optimization landscape, we prove a fundamental difference in learnability between classical and quantum models. For classical networks trained under cross entropy, we derive that as weight magnitudes diverge in Euclidean space ($\lambda \to \infty$), the trace of the FIM collapses exponentially: $\lim_{\lambda \to \infty} \text{Tr}(\text{FIM}) \to 0$. In sharp contrast, for QNNs, we utilize the properties of Haar integration to prove that the expected learning capacity is strictly confined within a stable regime, bounded away from zero (preventing saturation) and also bounded from above (preventing extreme learning sensitivity as chaos). Consequently, the corresponding theorem implies that the quantum learning capacity is invariant to weight inflation, ensuring the trace of the FIM never decays due to the growing weights.

\section{Robustness in deep reinforcement learning}

\begin{figure}[t]\centering
	\includegraphics[width=0.49\textwidth]{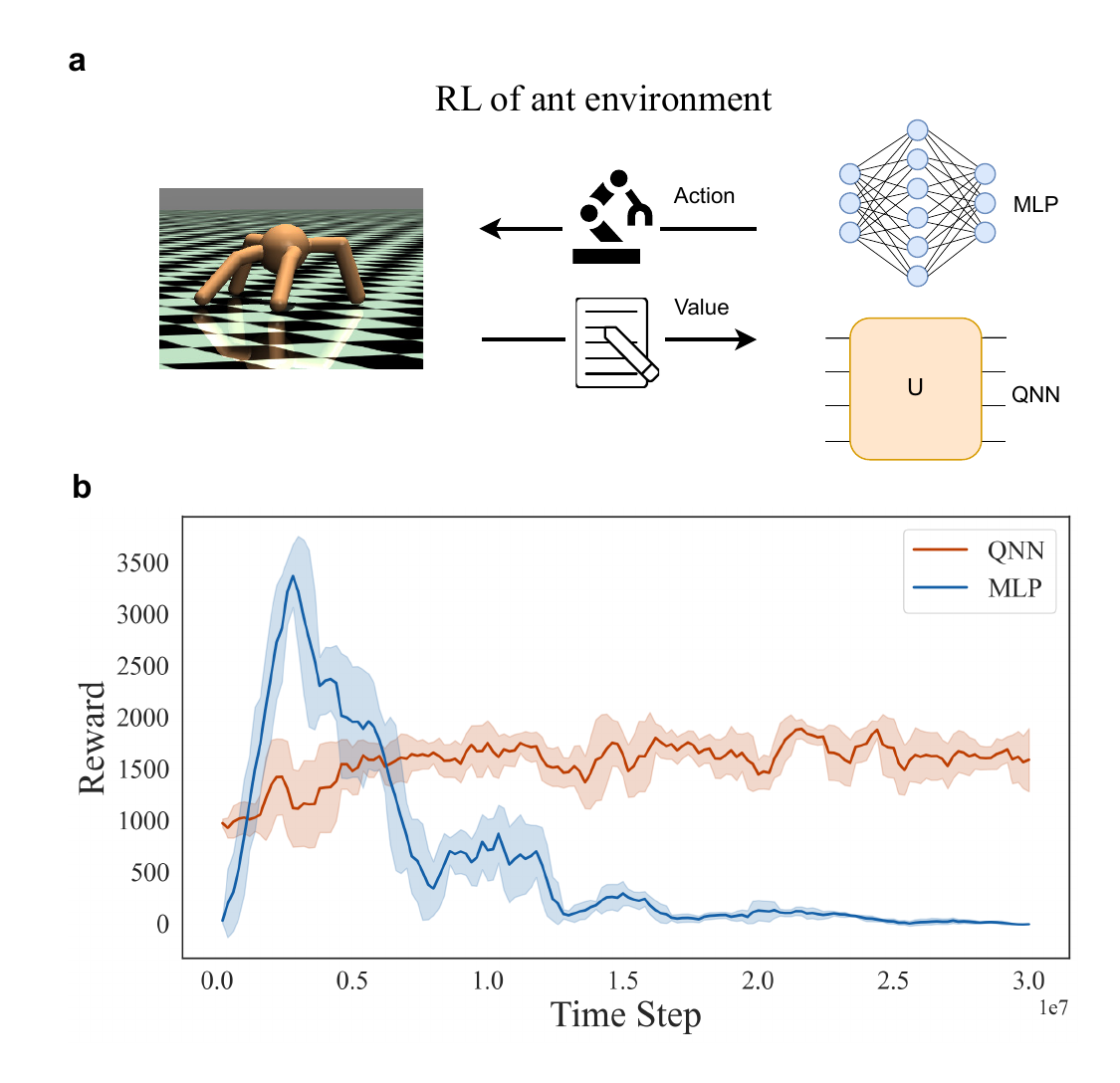}
\caption{\textbf{Quantum-enhanced agents maintain learning plasticity in a standard reinforcement learning benchmark.}
\textbf{a}, Schematic of the RL setup. An agent, using either a classical MLP or a quantum QNN as its function approximator for the policy and value networks, interacts with the Ant-v4 environment. The agent's goal is to learn a motion policy that maximizes reward for efficiently moving forward.
\textbf{b}, Evaluation reward (evaluated at each 200,000 steps) as a function of time step during training with PPO algorithm in a stationary environment. The agent with an MLP-based policy initially learns a strong policy, reaching a high reward. However, it subsequently suffers a catastrophic collapse in performance, demonstrating a severe loss of plasticity even without an explicit change of the task. Unlike the classical agent, the QNN-based agent learns more steadily and maintains a stable reward policy over the long term, showcasing its inherent ability to preserve plasticity in the dynamic context of RL. Curves show the moving average over 5-evaluation reward points window, with shaded areas representing the standard deviation of the underlying data within each window.}
\label{fig:rl}
\end{figure}

We next investigate whether the plasticity advantage extends to deep reinforcement learning (RL), where the data distribution shifts dynamically due to the agent's evolving policy. We employ the high-dimensional Ant-v4 environment (Fig.~\ref{fig:rl}\textbf{a}). In this continuous control task, the agent must coordinate four legs to maximize a composite reward signal, which incentivizes forward distance traveled and survival, while penalizing excessive control torques and contact forces.

We compare a standard PPO agent using MLP networks against a hybrid Quantum-PPO agent where the feature extractor is replaced by a 9-qubit, 12-layer quantum circuit with $SU(4)$ parameterized gates. The measurement probabilities of the quantum circuit are transformed by a scaled Tanh function to generate a dense feature vector, which is then projected via learnable linear heads to produce the stochastic action distribution and the scalar value estimate. Both agents are trained for 30 million timesteps across 16 parallel environments, demonstrating the remarkable scalability of our findings.

As shown in Fig.~\ref{fig:rl}\textbf{b}, the classical agent initially learns a high-performing policy but suffers a severe policy collapse in the later stages of training, a phenomenon identified as a symptom of plasticity loss in RL \cite{Nikishin2022rl, Abbas2023rlplasticity}. 
However, the quantum agent shows no such collapse. While its initial learning curve is more gradual, it maintains a stable, high-reward policy over much longer time steps. This result suggests that the intrinsic plasticity of quantum models acts as a natural regularizer against the instability of long-horizon RL, effectively stabilizing the policy optimization.

\section{Demonstration on quantum-native data}

\begin{figure}[t]\centering
	\includegraphics[width=0.5\textwidth]{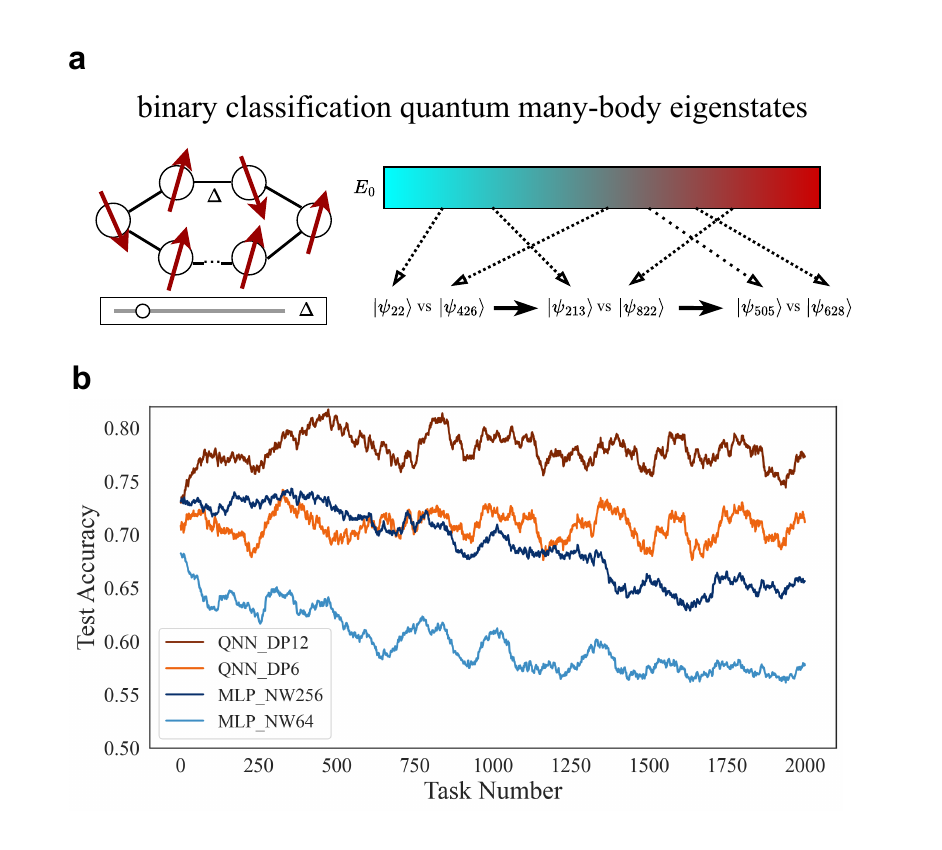}
	\caption{\textbf{QNNs exhibit superior plasticity on continual learning tasks with quantum native data.}
\textbf{a}, Schematic of the quantum native continual learning task. The dataset consists of many-body eigenstates from the one-dimensional XXZ Hamiltonian with periodic boundary conditions, generated by varying the anisotropy parameter $\Delta$. Each task in the continual learning sequence is a binary classification problem to distinguish between two sets of randomly sampled eigenstates from the full energy spectrum: the $i$-th eigenstate $|\psi_i\rangle$ and the $j$-th eigenstate $|\psi_j\rangle$.
\textbf{b}, Test accuracy as a function of task number for QNNs and MLPs of varying model sizes. QNNs, particularly the deeper variant (DP12), successfully learn and maintain a high classification accuracy across the sequence of 2,000 tasks. In contrast, the MLPs show a clear degradation in performance, indicating that the loss of plasticity is a fundamental issue for classical models even for data with inherent quantum structure. This result demonstrates the robust plasticity of QNNs in their native domain. Curves show the moving average over 100-task window.}
\label{fig:qdata}
\end{figure}

Finally, we analyze the performance on a quantum-native continual learning task to confirm the universality of the observed advantages. We generate a sequence of 2,000 binary classification tasks based on the eigenstates of a one-dimensional XXZ spin chain with varying anisotropy parameter $\Delta$ (Fig.~\ref{fig:qdata}\textbf{a}). For this task, the QNN naturally operates in the same Hilbert space as the data input, with predictions derived from the expectation value $\langle Z \rangle$ of the first qubit.

Fig.~\ref{fig:qdata}\textbf{b} reveals a profound result: classical MLPs still lose plasticity on the quantum dataset, which extends our findings to different data modalities. 
In contrast, QNNs demonstrate a dual advantage unique to the quantum-native regime. First, it learns better: the deeper QNN variant achieves significantly higher accuracy than the MLP, benefiting from the inductive bias alignment between the quantum model and the quantum data structure. Second, it learns longer: unlike the MLP, which degrades over time, the QNN preserves this high performance with zero degradation over 2,000 tasks. The synergy between superior representational power and indefinite plasticity positions QML as a unique powerful paradigm for processing varying quantum information streams.

\section{Discussion}

Our findings suggest that the loss of plasticity is a pervading failure mode in deep learning due to the training in unbounded Euclidean spaces. In classical networks, the continuous accumulation of weight updates drives the model into saturation regimes where the effective rank of the representation collapses. Our counterfactual experiments with Sin-MLPs reveal that simply bounding the activation function is insufficient for classical learning models to maintain plasticity. The quantum ansatz, by virtue of its unitary evolution and periodic parameterization on a compact manifold, naturally enforces the parameter space constraint. This nature ensures that the optimization landscape remains navigable indefinitely as evidenced by the stable test accuracy and Fisher information spectrum during CL. Analytically, this contrast is detailed in Appendix, where we define and rigorously prove the asymptotic limits of the Fisher information spectrum for both architectures. It is important to emphasize that this dynamic advantage is conditional on initial trainability. Once in a trainable regime, the capacity does not dynamically decay.

While it is well established that quantum parameters reside on a compact periodic manifold, discovering a practical problem that fundamentally benefits from this geometric property has remained elusive. Our central innovation lies in bridging this gap. We identify the loss of plasticity in continual learning as the exact real-world scenario where this quantum constraint transforms into a decisive advantage. Thus, we demonstrate that quantum unitarity is not merely a physical rule, but a structural solution to a fundamental pathology in artificial intelligence.

We remark that within the classical deep learning community, algorithmic interventions such as regenerative regularization \cite{Kumar2023reg}, new parameter injection \cite{Nikishin2023}, and selective weight resetting (e.g., continual backpropagation) \cite{Dohare2024} have been proposed to mitigate plasticity loss. However, these approaches represent extrinsic patches: they rely on heuristic modifications to the training objective or invasive surgeries on the network state to counteract the model's natural tendency toward rigidity. Such methods often introduce sensitive hyperparameters and can disrupt prior knowledge. The significance of our work lies in demonstrating that QNNs do not require these extrinsic interventions. Their plasticity is intrinsic, emerging directly from the QNN nature. This distinguishes the quantum approach as a first-principle structural solution rather than an engineering workaround, offering a more robust foundation for CL.

This insight suggests a redefinition of the strategic utility of QML. Historically, the field has been mainly focused on the search for computational speedups, typically quantified by asymptotic complexity. Our work instead identifies a distinct and perhaps more immediate axis of quantum advantage: \textit{learning robustness}. While QNNs may not always train faster than classical counterparts, they learn longer and more reliably in non-stationary environments. This finding complements recent research demonstrating quantum advantages in \textit{learning resilience} against data noise and unlearning efficiency \cite{Chen2025resilience}, collectively suggesting that the geometric properties of quantum learning models offer a superior basis for adaptability. With learning robustness and resilience revealed, the evaluation on QML shifts the metric from mere \textit{time-to-solution} to \textit{quality-of-solution} over the model's entire lifecycle. This distinction is critical for real-time control systems and sustainable AI, where the train-reset-retrain cycle is ecologically prohibitive. Furthermore, on quantum-native tasks, we observe a dual advantage in terms of superior representational alignment and indefinite plasticity, which make QML the viable path for processing continuous streams of quantum information without degradation.

We conclude by noting the limitations and future directions of this research. Our study is conducted via exact simulations using {\sf TensorCircuit-NG} \cite{Zhang2022tc}. The effect of quantum noise on continual quantum learning is worth further investigation.
Besides, it is also interesting to explore whether we can design efficient quantum-inspired classical learning gadgets capable of maintaining plasticity like quantum models.
Finally, we emphasize that while the quantum approach intrinsically preserves plasticity, it does not inherently solve the complementary challenge of stability. Both quantum and classical models in our study exhibit catastrophic forgetting of previous tasks when trained without explicit stability-preserving regularizers. A complete continual learning system would require augmenting the plastic quantum backbone with specialized stabilization techniques or memory mechanisms. This hybrid approach represents a promising frontier for realizing better AI that is both sufficiently stable to remember and plastic enough to evolve.

\section*{Acknowledgments}
This work was supported by Quantum Science
and Technology-National Science and Technology Major Project (No. 2024ZD0301700), the National Natural Science Foundation of China (Nos. 12504599 and 12574546), and NSAF (No. U2330401).

\section*{Data Availability}
The data that support the findings of this article are openly available at \url{https://github.com/sxzgroup/quantum-plasticity}.

\appendix

\section{Methods and Pipelines}

\subsection{Experiment 1: Supervised Continual Learning on Permuted MNIST}

\textbf{Task Formulation.} We simulate a domain change learning scenario using the standard MNIST dataset. The experiment consists of a sequence of $T=1,000$ independent classification tasks. For each task $t$, a fixed random permutation matrix $P_t$ is generated. The raw input images (28$\times$28 pixels) are flattened into vectors $\mathbf{x} \in \mathbb{R}^{784}$. For task $t$, the input provided to the model is $\mathbf{x}'_t = P_t \mathbf{x}$. The 10-class target labels $y \in \{0, \dots, 9\}$ remain unchanged.

\textbf{Data Preprocessing for Quantum Models.}
The inputs are zero-padded to map the 784-dimensional vector to a $N$-qubit Hilbert space ($d=2^N$). We utilize $N=10$ qubits ($d=1024$). The 784-pixel vector is padded with zeros to form a vector of size 1024. This vector is then $L_2$-normalized to unit magnitude to serve as a valid quantum state vector for amplitude encoding.

\textbf{Model Architectures.}
\begin{itemize}
    \item \textbf{Classical MLP:} The network consists of an input layer ($d=784$), a single hidden Dense layer with ReLU activation, and a Dense output layer with 10 units using Softmax activation head. To test the impact of capacity, we vary the width of hidden layer  (Network Width, NW) across $\{4, 6, 10, 16\}$ neurons.
    \item \textbf{Deep QNN:} We employ a VQC on $N=10$ qubits (Fig.~\ref{fig:circuit}). The input is loaded via amplitude encoding. The variational ansatz includes $L$ layers (Depth, DP) blocks. Each layer comprises general two-qubit $SU(4)$ unitary gates applied to alternating pairs of neighboring qubits (even-odd pairs followed by odd-even pairs in brickwall pattern). We vary the circuit depth $\text{DP} \in \{2, 4, 6, 9, 16\}$ to evaluate scalability.
    \item \textbf{QNN Readout:} A computational basis measurement is performed on all 10 qubits. The logarithmic probabilities of the first 10 bitstrings are interpreted as logits. These logits are processed by a softmax activation and fed into the cross-entropy loss.
\end{itemize}

To maximize local expressivity and avoid architectural bottlenecks, the fundamental building block of our variational ansatz is the general parameterized two-qubit gate, $V \in SU(4)$. Unlike hardware-efficient ansatzes that rely on restricted gate sets (e.g., $R_y, R_z, CNOT$), a general two-qubit unitary possesses 15 degrees of freedom (ignoring global phase), allowing it to reach any state in the two-qubit Hilbert space.

We parameterize this gate using the canonical exponential map of the $\mathfrak{su}(4)$ Lie algebra. Let $\sigma_0 = I$ and $\{\sigma_1, \sigma_2, \sigma_3\} = \{X, Y, Z\}$ denote the standard Pauli basis. The gate $V(\boldsymbol{\theta})$ is defined by 15 trainable real parameters $\boldsymbol{\theta}$ as \cite{Dilip2022}:
\begin{equation}
    V(\boldsymbol{\theta}) = \exp\left( -\frac{i}{2} \sum_{\substack{\alpha, \beta \in \{0,1,2,3\} \\ (\alpha, \beta) \neq (0,0)}} \theta_{\alpha,\beta} \, (\sigma_\alpha \otimes \sigma_\beta) \right).
\end{equation}
The summation runs over all 15 non-identity tensor products of Pauli matrices. By arranging these gates in a dense brickwall or ladder topology, the network creates a highly entangled state with a parameter manifold that is both compact and maximally expressive for its depth.

\textbf{Training Protocol.} All models are trained using the Adam optimizer. The batch size is 128 for all models.

\begin{itemize}
    \item \textbf{Classical:} Learning rate $0.001$, trained for 1 epoch per task.
    \item \textbf{Quantum:} Learning rate $0.03$ trained for 5 epochs per task.
\end{itemize}

\textbf{Plasticity Metrics.} We also evaluate the average $L_2$ norm of the weight matrices (excluding biases for MLPs) and the average $L_2$ norm of the gradients computed on a fixed batch of data during training for each task.

\subsection{Experiment 2: Supervised Continual Learning on Split CIFAR-100}

\textbf{Task Formulation.} We construct a sequence of $T=3,000$ binary classification tasks using the CIFAR-100 dataset. For each task, two distinct classes are sampled uniformly at random from the 100 available classes. The training and test datasets are filtered to include only samples from these two classes, and labels are remapped to $\{0, 1\}$.

\textbf{Data Preprocessing.}
The $32\times32\times3$ RGB images are firstly converted to grayscale ($32\times32\times1$) and then flattened into vectors of dimension $d=1024$ for both quantum and classical models.
For both classical and quantum models, pixel values are $L_2$-normalized to unit magnitude. This ensures that the classical model operates on the same geometric constraints as the quantum model for the input.

\textbf{Model Architectures.}
\begin{itemize}
    \item \textbf{Classical MLP:} A deeper architecture is built for this visual task. It consists of an input layer ($d=1024$), a first hidden layer with 256 ReLU neurons, a second hidden layer with varied ReLU neurons as $\text{NW} \in \{128, 256\}$, and a linear Dense output head for logits.
    \item \textbf{Deep QNN:} We use the same $N=10$ qubit system as in Experiment 1, but with significantly increased circuit depths to handle real-world image complexity. We test circuit depths $\text{DP} \in \{10, 20, 30\}$.
    \item \textbf{QNN Readout:} We measure the expectation value of the Pauli-$Z$ operator on the last qubit as the logit, which is further used for classification loss. 
\end{itemize}

\textbf{Training Protocol.} Both models are trained using the Adam optimizer with a fixed learning rate of $\eta=0.001$ for 5 epochs per task. Batch size is 128.

\textbf{Fisher Information Analysis.} To quantify the effective learnability, we compute the trace of the FIM. To ensure temporal consistency, we define a fixed probe dataset $D_{\text{probe}}$ comprising 64 randomly sampled images from the first task. The FIM trace is approximated as the expected squared $L_2$ norm of the score function (log-likelihood gradient) \cite{Achille2017}:
\begin{equation}
\text{Tr}(F(\boldsymbol{\theta}_t)) \approx \frac{1}{|D_{\text{probe}}|} \sum_{(\mathbf{x}, y) \in D_{probe}} \left\| \nabla_{\boldsymbol{\theta}} \log p(y|\mathbf{x}; \boldsymbol{\theta}_t) \right\|_2^2
\end{equation}

\begin{figure}[t]\centering
	\includegraphics[width=0.5\textwidth]{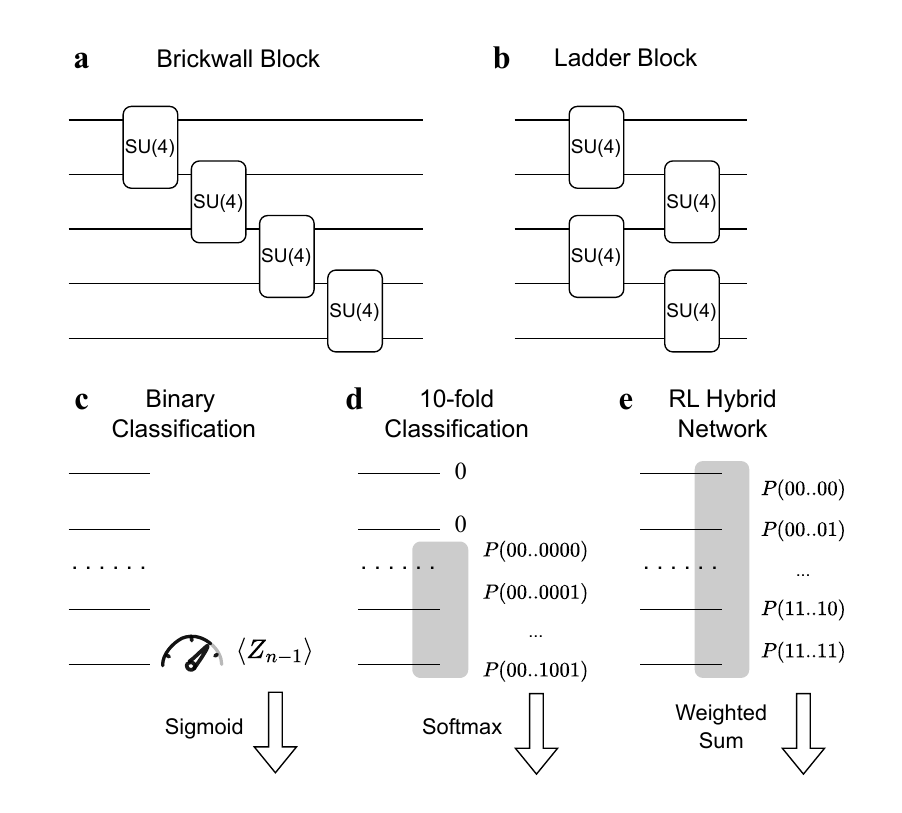}
	\caption{\textbf{Quantum neural network architectures and readout strategies employed in different experimental settings.} 
\textbf{a}, \textbf{b}, Schematics of the variational ansatz blocks used to construct the deep QNNs. Each block is composed of parameterized two-qubit general unitary gates, $SU(4)$, acting on nearest neighbor qubits. 
\textbf{a}, The \textbf{Brickwall Block} architecture applies gates to alternating pairs of qubits.
\textbf{b}, The \textbf{Ladder Block} architecture applies the same parameterized $SU(4)$ gates in a sequential descending pattern.
\textbf{c}--\textbf{e}, Task-specific readout protocols mapping quantum states to classical predictions. 
\textbf{c}, \textbf{Binary Classification Readout} (used for Split CIFAR-100 and Quantum-Native tasks): The expectation value $\langle Z \rangle$ of the given qubit is measured and passed through a sigmoid activation to yield a binary class probability. 
\textbf{d}, \textbf{10-class Classification Readout} (used for Permuted MNIST): The probabilities of the first 10 computational basis states ($P(...0000), \dots, P(....1001)$) are extracted from the output distribution and normalized via a Softmax function to produce the categorical prediction. 
\textbf{e}, \textbf{Hybrid RL Readout} (used for RL): The full probability distribution over all $2^N$ basis states serves as a dense feature vector. This vector is processed via a learnable linear projection (weighted sum) to generate the continuous policy actions and scalar value estimates, respectively.}
\label{fig:circuit}
\end{figure}

\subsection{Experiment 3: Deep Reinforcement Learning}

\textbf{Environment.} We employ the {\sf Ant-v4} environment from the {\sf MuJoCo} physics engine via the {\sf Gymnasium} library. 
The observation space is 27-dimensional, comprising joint angles, joint velocities, and center-of-mass coordinates (excluding external contact forces).

\textbf{Reward Function.} The agent is trained to maximize a composite reward signal $R$ calculated at each timestep:
\begin{equation}
R = \underbrace{v_x}_{\text{Forward}} + \underbrace{1.0}_{\text{Healthy}} - \underbrace{0.5 \|\mathbf{a}\|^2}_{\text{Control Cost}} - \underbrace{5\times 10^{-4} \|\mathbf{F}_{clip}\|^2}_{\text{Contact Cost}}
\end{equation}
where $v_x$ is the forward velocity, $1.0$ is a survival bonus if the ant does not fall, $\mathbf{a}$ is the action vector (control torque), and $\mathbf{F}_{clip}$ represents the external contact forces clipped to the range $[-1, 1]$. Note that while contact forces are excluded from the agent's input observation (to test plasticity on kinematic features), they are actively penalized in the reward to encourage efficient locomotion.

\textbf{Data Preprocessing for Quantum Models.} The 27-dimensional observation vector is zero-padded to 9-qubit dimension $d=512$ ($2^9$) and $L_2$-normalized to be compatible with quantum amplitude encoding.

\textbf{Agent Architectures.} We utilize the PPO algorithm with separate policy ($\pi$) and value ($V$) networks.
\begin{itemize}
    \item \textbf{Classical Agent:} The policy network comprises two hidden Dense layers with 1024 and 256 units (ReLU activation), followed by a linear action head of 8 outputs. The value network comprises two hidden layers with 256 units each.
    \item \textbf{Hybrid Quantum Agent:} The observation vector from the environment is first fed into a fixed-structure QNN acting on $N=9$ qubits. The ansatz consists of 12 layers of $SU(4)$ gates in ladder layout.
    \item \textbf{Quantum-Classical Interface:} To map the quantum state to classical policy features, we measure the probability distribution $p(\mathbf{z})$ of the basis states. These probabilities are processed by a scaled nonlinear activation function $h(\mathbf{z}) = \tanh(5.0 \; p(\mathbf{z}))$ to generate a dense feature vector. This vector is then projected via learnable linear heads (with no additional classical layers) to produce the stochastic action distribution vector and scalar value estimate.
\end{itemize}

\textbf{Training Protocol.} Agents are trained for a total of 30 million timesteps using 16 parallel environments. We use a rollout buffer of 2048 steps, a batch size of 128, and 10 optimization epochs per update. The learning rate is set to $10^{-4}$ for the classical agent and $0.002$ for the quantum agent. 

\subsection{Experiment 4: Continual Learning on Quantum Data}

\textbf{Task Formulation.} We generate a dataset of many-body quantum states derived from the one-dimensional Heisenberg XXZ Hamiltonian on a ring of $L=10$ spins:
\begin{equation}
H = \sum_{i=1}^{L} (X_i X_{i+1} + Y_i Y_{i+1} + \Delta Z_i Z_{i+1})
\end{equation}
Periodic boundary conditions are assumed. We vary the anisotropy parameter $\Delta$ from $-2.0$ to $2$ in steps of $0.02$. For each $\Delta$, we compute the full eigenspectrum. We construct a sequence of 2,000 tasks; each task involved binary classification between two distinct eigenstates $|\psi_i\rangle$ and $|\psi_j\rangle$ sampled uniformly from the spectrum. In other words, each task consists of 200 data points, each data point is a many-body wavefunction from either the $i$-th (label 0) or $j$-th (label 1) eigenstates of Hamiltonian with certain $\Delta$.

\textbf{Input Representations.}
\begin{itemize}
    \item \textbf{Classical Input:} Since standard MLPs cannot process complex amplitudes directly, the state vector $|\psi\rangle \in \mathbb{C}^{1024}$ is decomposed. The real and imaginary parts are concatenated to form a real-valued input vector $\mathbf{x} \in \mathbb{R}^{2048}$.
    \item \textbf{Quantum Input:} The state $|\psi\rangle$ is loaded directly into the quantum register as the initial wavefunction, employing the native Hilbert space structure.
\end{itemize}

\textbf{Model Architectures.}
\begin{itemize}
    \item \textbf{Classical MLP:} A feed-forward network with two hidden layers of $NW$ neurons each (ReLU activation) and a single unit output with Sigmoid activation. We test $NW=64, 256$.
    \item \textbf{Quantum Classifier:} A VQC with $DP$ blocks of $SU(4)$ two-qubit gates in one dimensional ladder pattern on 10 qubits. We test $DP=6, 12$. The prediction logit is obtained by measuring the expectation value $\langle Z_0 \rangle$ of the first qubit which further goes through the sigmoid activation in the binary cross entropy loss.
\end{itemize}

\textbf{Training Protocol.} Both models are trained for 5 epochs per task with a batch size of 64 and a learning rate of $0.002$ using the Adam optimizer.

\section{Analytical Derivation of Fisher Information Spectrum}

In this section, we derive the asymptotic behavior of the FIM for classical and quantum neural networks trained with the binary cross entropy (BCE) loss. We show that for classical networks (regardless of activation function), the unbounded growth of weight magnitudes leads to a collapse of the FIM trace, corresponding to a loss of plasticity. In contrast, we prove that the unitary nature of QNN guarantees a stable, non-vanishing FIM spectrum.

\subsection{General Form of FIM for Diverse Activation and Loss Functions}

Consider a classification problem with input $\mathbf{x}$ and target $y$. Let the model produce a logit vector $\mathbf{f}(\mathbf{x}; \boldsymbol{\theta})$, which is transformed into a predictive distribution $p_j = \text{softmax}(\mathbf{f})_j$ (for multi-class) or $p = \sigma(f)$ (for binary). 

By definition, the model Fisher Information Matrix is the expected outer product of the log-likelihood gradients: $F(\boldsymbol{\theta}) = \mathbb{E}_{\mathbf{x}, y \sim p_\theta}[\nabla_{\boldsymbol{\theta}} \mathcal{L} \nabla_{\boldsymbol{\theta}} \mathcal{L}^T]$. Applying the chain rule through the pre-activation logits $\mathbf{f}(\mathbf{x}; \boldsymbol{\theta})$, and exploiting the equivalence between the FIM and the Generalized Gauss-Newton matrix, the expected second-order derivative of the model dynamics strictly vanishes under the model's own predictive distribution. This mathematically decouples the FIM into two distinct geometric components: the curvature of the output distribution and the Jacobian of the parameter manifold.

Consequently, the FIM reduces to the expectation over the data distribution:
\begin{equation}
    F(\boldsymbol{\theta}) = \mathbb{E}_{\mathbf{x}} \left[ \sum_{i,j} H_{ij} \nabla_{\boldsymbol{\theta}} f_i(\mathbf{x}) \nabla_{\boldsymbol{\theta}} f_j(\mathbf{x})^T \right],
    \label{eq:general_fim_matrix}
\end{equation}
where matrix $H_{ij} = \frac{\partial^2 \mathcal{L}}{\partial f_i \partial f_j}$. 
For the binary case (Sigmoid), $H$ reduces to the scalar $\xi = p(1-p)$. For the multi-class case (Softmax), the components are $H_{ij} = p_i \delta_{ij} - p_i p_j$, where $p_i$ is the probability of class $i$.
In both cases, $\text{Tr}(F)$ represents the total sensitivity of the model's output distribution to parameter perturbations.

\subsection{Plasticity Collapse in Classical Networks}

We analyze the behavior of Eq. \eqref{eq:general_fim_matrix} under the condition of weight growth, a phenomenon ubiquitously observed in continual learning as shown in main text.

\begin{proposition}[Asymptotic Saturation]
Consider a classical neural network $f(\mathbf{x}; \mathbf{w}) = \mathbf{w}_{out}^T \mathbf{h}(\mathbf{x}; \mathbf{w}_{int})$, where $\mathbf{w}_{out}$ are the readout weights and $\mathbf{h}$ is the latent space representation. If the magnitude of the weights grows unboundedly, i.e., $||\mathbf{w}|| \to \infty$, then $\text{Tr}(F(\mathbf{w})) \to 0$.
\end{proposition}

In continual learning, to minimize loss on past data, the optimizer increases the magnitude of the weight vector. Let us parameterize this growth by a scaling factor $\lambda$ as $\mathbf{w} \to \lambda \mathbf{w}_0$ when $\lambda \to \infty$. 

\textbf{Case 1: Sigmoid Binary Head.}
The logit output scales as $f_\lambda(\mathbf{x}) \approx \lambda f_0(\mathbf{x})$. The curvature term $\xi(\lambda) = \sigma(\lambda f_0) (1 - \sigma(\lambda f_0))$ decays exponentially: $\lim_{\lambda \to \infty} \xi(\lambda) \approx e^{-\lambda |f_0|}$. Since the gradients $\nabla_{\mathbf{w}} f$ scale at most polynomially with $\lambda$ (e.g., linearly for ReLU or stay bounded for bounded activations), the product $\xi(\lambda) ||\nabla f||^2$ vanishes as $\lambda \to \infty$.

\textbf{Case 2: Softmax Multi-class Head.}
In classical classification architectures, regardless of the internal hidden activations (e.g., ReLU, Tanh, or periodic Sine), the final mapping to the unnormalized logits is a linear readout layer: $f_i(\mathbf{x}) = \mathbf{w}_{out, i}^T \mathbf{h}(\mathbf{x})$, where $\mathbf{h}(\mathbf{x})$ is the representation from the penultimate layer. In continual learning, unbounded parameter growth implies $||\mathbf{w}_{out}|| \to \infty$. Let us parameterize the divergent readout weights as $\mathbf{w}_{out}(\lambda) = \lambda \mathbf{w}_{out, 0}$ for $\lambda \to \infty$. Even if $\mathbf{h}(\mathbf{x})$ is strictly bounded by specific activation functions, the output logits scale linearly: $f_{\lambda, i}(\mathbf{x}) \approx \lambda f_{0, i}(\mathbf{x})$. For any given input $\mathbf{x}$, let $k = \text{argmax}_j f_{0, j}$ be the strictly dominant class, and let $\Delta_{min} = \min_{j \neq k} (f_{0, k} - f_{0, j}) > 0$ be the minimum margin in the initial logit space. As $\lambda \to \infty$, the probability of the dominant class $p_k(\lambda)$ approaches $1$, and any non-dominant class $j \neq k$ decays exponentially:\begin{equation}p_j(\lambda) = \frac{e^{\lambda f_{0, j}}}{\sum_i e^{\lambda f_{0, i}}} \le e^{-\lambda (f_{0, k} - f_{0, j})} \le e^{-\lambda \Delta_{min}}.\end{equation}Consequently, every element in the stiffness matrix $H_{ij}(\lambda) = p_i \delta_{ij} - p_i p_j$ becomes bounded by this exponential decay. Specifically, $H_{kk} = p_k(1-p_k) \approx \mathcal{O}(e^{-\lambda \Delta_{min}})$, and $H_{jj} \approx \mathcal{O}(e^{-\lambda \Delta_{min}})$. Therefore, the matrix norm scales as:\begin{equation}||H(\lambda)|| \sim \mathcal{O}(e^{-\lambda \Delta_{min}}).\end{equation}

Crucially, the ultimate FIM trace is the product of this stiffness matrix $H$ and the squared Jacobian $||\nabla_{\mathbf{w}} \mathbf{f}||^2$. By the chain rule, the gradient norm with respect to the network weights grows at most polynomially. For an $L$-layer network, this growth is bounded by $\mathcal{O}(\lambda^{2L-2})$, regardless of whether the internal activations are unbounded (e.g., ReLU) or strictly bounded (e.g., Sine or Tanh), because the gradients are inevitably multiplied by the divergent subsequent weight matrices. Combining these two dynamics, the asymptotic limit of the FIM trace becomes a fundamental competition between polynomial gradient growth and exponential curvature collapse, leading to the vanishing FIM. This rigorously establishes that the geometric unboundedness of the Euclidean parameter space inevitably extinguishes the model's plasticity, a structural pathology that cannot be circumvented by simply tuning activation functions.

\subsection{Preservation of Plasticity in Quantum Neural Networks}

We now prove that QNNs are immune to this saturation mechanism due to the compactness of their parameter manifold.
\begin{definition}[QNN Setup and Ansatz Structure]
We define the QNN function $f(\mathbf{x}; \boldsymbol{\theta})$ as the expectation value of a measurement observable $O$ (typically a Pauli string):
\begin{equation}
    f(\mathbf{x}; \boldsymbol{\theta}) = s \cdot \langle 0 | U^\dagger(\mathbf{x}, \boldsymbol{\theta}) O U(\mathbf{x}, \boldsymbol{\theta}) | 0 \rangle.
\end{equation}
Here, $s$ is a fixed, non-trainable scalar. The unitary $U(\mathbf{x}, \boldsymbol{\theta})$ is composed of interleaved data encoding layers $S(\mathbf{x})$ and trainable layers $W(\boldsymbol{\theta})$. The trainable unitary is parameterized by rotation angles $\boldsymbol{\theta} = \{\theta_k\}_{k=1}^M$ generated by Pauli strings $H_k \in \{I, X, Y, Z\}^{\otimes N}$:
\begin{equation}
    W(\boldsymbol{\theta}) = \prod_{k=1}^M e^{-i \theta_k H_k}.
\end{equation}
Because $e^{-i \theta_k H_k}$ is periodic, the effective parameter space is the $M$-dimensional torus, $\mathbb{T}^M = [0, 2\pi)^M$.
\end{definition}

\begin{theorem}[Bounded FIM]
For a QNN defined above, the trace of the FIM under BCE loss satisfies three conditions:
\begin{enumerate}
    \item \textbf{Non-Vanishing:} The curvature term $\xi = p(1-p)$ is strictly lower-bounded.
    \item \textbf{Non-Explosion:} The gradient norm is strictly upper-bounded, preventing chaotic landscapes.
    \item \textbf{Recurrence:} The expected FIM trace over the parameter manifold is a non-zero constant.
\end{enumerate}
\end{theorem}

\begin{proof}
Recall the FIM trace decomposition:
$$ \text{Tr}(F) = \mathbb{E}_{\mathbf{x}} \left[ \underbrace{p(1-p)}_{\xi} \times \underbrace{\sum_{k=1}^M (\partial_k f)^2}_{||\nabla f||^2} \right] $$

\textbf{Part 1: Lower Bound on Curvature.}
In classical networks, $|f| \to \infty$ as weights grow. In QNNs, the output is the expectation value of a Pauli observable $O$ with spectral norm $||O||_\infty = 1$. Thus, the logit is strictly bounded for all $\boldsymbol{\theta}$:
\begin{equation}
    |f(\mathbf{x}; \boldsymbol{\theta})| = |s \cdot \langle \psi | O | \psi \rangle| \le |s| \cdot ||O||_\infty = |s|.
\end{equation}
Since $|f| \le |s|$, the sigmoid derivative is strictly lower-bounded:
\begin{equation}
    1/4\geq\xi = \sigma(f)(1 - \sigma(f)) \ge \sigma(|s|)(1 - \sigma(|s|)) = C_{min} > 0.
\end{equation}

\textbf{Part 2: Upper Bound on Gradients.}
In QNNs, we can bound the gradient using the parameter-shift rule. For Pauli generators with eigenvalues $\pm 1$, the gradient is exact:
\begin{equation}
    \partial_k f = s \cdot \frac{1}{2} \left( \langle H_k \rangle_{\theta_k+\pi/2} - \langle H_k \rangle_{\theta_k-\pi/2} \right).
\end{equation}
Since the expectation value $\langle H_k \rangle$ is bounded by $[-1, 1]$, the gradient is strictly bounded by the scaling factor:
\begin{equation}
    |\partial_k f| \le |s| \cdot \frac{1}{2} \cdot (1+1) = |s|.
\end{equation}
Consequently, the total FIM trace has a strict global upper bound for any configuration of parameters:
\begin{equation}
    \text{Tr}(F) \le \underbrace{0.25}_{\max(\xi)} \cdot \sum_{k=1}^M (|s|)^2 = 0.25 \cdot M \cdot s^2 = C_{max} < \infty.
\end{equation}
This proves that the QNN landscape cannot become infinitely rugged or chaotic, i.e. extremely sensitive to training data.

\textbf{Part 3: Stable Average Capacity.}
With the FIM strictly upper bounded, we finally show it does not statistically vanish with growing weights $\lambda$ in CL. As parameters $\boldsymbol{\theta}$ traverse the compact torus $\mathbb{T}^M$, we examine the expected learning capacity.
According to results from concentration of measure in quantum circuits, the expected squared gradient for a Pauli generator scales inversely with the effective dimension of the Hilbert space $D_{\text{eff}}$, where $D_{\text{eff}} \propto 2^N$ for sufficient expressive ansatzes or global observables and $D_{\text{eff}} \propto \text{poly}(N)$ for barren plateau mitigated cases \cite{McClean2018bp}:
\begin{equation}
    \mathbb{E}_{\boldsymbol{\theta}} [(\partial_k f)^2]  \propto \frac{1}{D_{\text{eff}}}.
\end{equation}
Although this value depends on the system size or Hilbert space dimension, it is a static constant determined solely by the architecture and doesn't change over the training time. Unlike classical networks where the FIM decays exponentially with weight growth ($e^{-\lambda} \to 0$), the quantum gradient expectation remains invariant to the duration of training or $\lambda$.
Combining Parts 1, 2, and 3, we establish that the QNN plasticity is permanently confined to a stable, active regime:
\begin{equation}
    0 < \frac{C_{min}}{D_{\text{eff}}} \le \mathbb{E}[\text{Tr}(F)] \le C_{max} < \infty.
\end{equation}
\end{proof}

Note that the analytical derivation above specifically addresses binary classification tasks with the expectation value readout strategy, where the model output is defined as the expectation of a bounded observable. This formulation corresponds directly to the experimental setups employed in Split CIFAR-100 and Quantum Data classification.

\subsection{Contrast with barren plateaus and dynamic plasticity}
It is essential to distinguish the static challenge of barren plateaus (BP) from the dynamic loss of plasticity. BP is an initialization hurdle: as the number of qubits $n$ scales, the variance of gradients (and thus $\text{Tr}(F)$) vanishes exponentially as $1/2^n$ for generic circuits and random initialization. However, QNNs utilizing amplitude encoding can represent high-dimensional features of length $D$ using a logarithmic number of qubits $n = \log_2 D$. In our case, the 1,024-dimensional visual features from CIFAR-100 require only 10 qubits ($2^{10} = 1,024$). While the gradient variance in such circuits theoretically scales as $1/2^n = 1/D$, the dimensionality $D$ in real-world data is not asymptotically exponentially large. This ensures that the system possesses a high information density that maintains expressive power while keeping the gradient signal at a manageable floor of $1/D$. Since $D$ corresponds to fixed real-world input lengths, this gradient signal remains robustly within the trainable region despite the circuit's depth.

In contrast, loss of plasticity is a training-induced decay where the capacity collapses due to parameter drift into infinity. Our results show that for a QNN that has survived the BP regime (either through the use of high information density, local observables, specific circuit structures, or smart initialization), its learning capacity $1/D_{\text{eff}}$ remains a non-zero static constant. While the gradient signal may be small (scaling as $1/D$), it does not decay further as more tasks are learned.

To empirically verify this stability, we conducted a barren plateau analysis on the Split CIFAR-100 task setup. We varied the circuit depth from 1 to 60 and measured the squared norm of the gradient at initialization, averaged over 32 random seeds. As shown in Fig.~\ref{fig:bp_analysis}, the gradient signal $||\nabla_\theta \mathcal{L}||^2$ remains robust and non-vanishing (averaging $\sim 5 \times 10^{-2}$) across all depths, including our standard study depth of 30. The model resides in a robustly trainable regime. This is fundamentally different from classical models, where the FIM trace decays exponentially towards zero over the task sequence, regardless of the initial capacity.

\begin{figure}[h]
    \centering
    \includegraphics[width=0.45\textwidth]{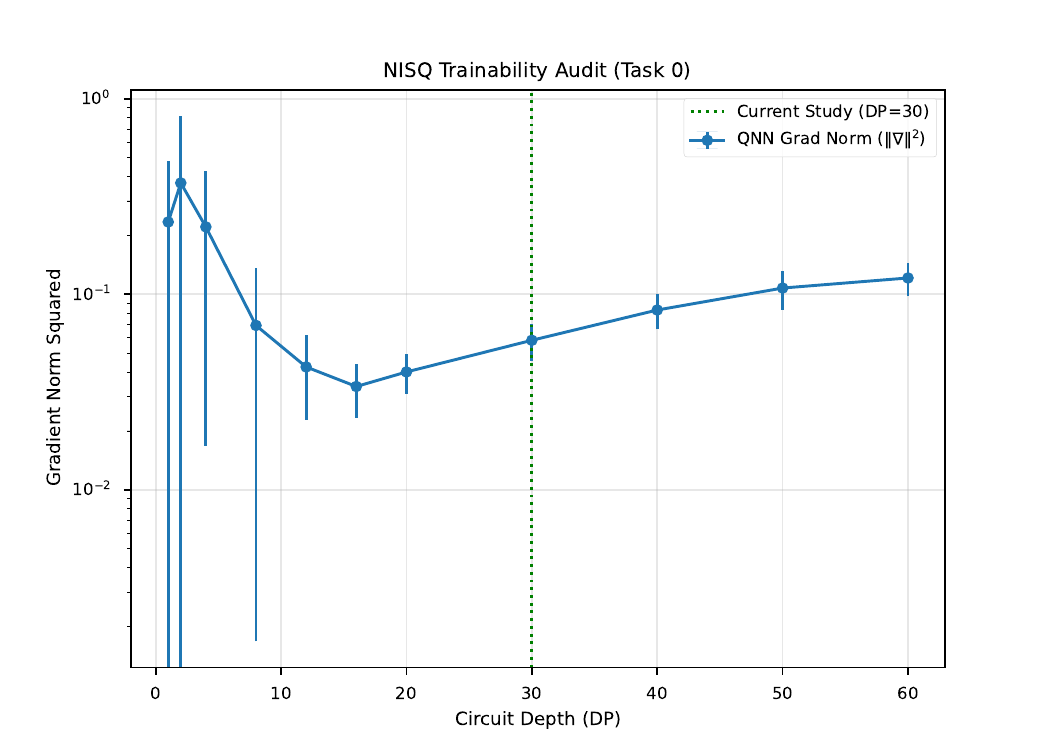}
    \caption{\textbf{Trainability of QNN for different depths.} The squared norm of the gradient is plotted as a function of the variational circuit depth for the first task of the Split CIFAR-100 benchmark. The stable, non-vanishing gradient floor empirically confirms that the QNN operates in a trainable regime, ensuring that the tracked plasticity metrics are not artifacts of initial untrainability.}
    \label{fig:bp_analysis}
\end{figure}

\section{Generality of Plasticity Preservation across Quantum Circuit Ansatzes}

To assess the universality of our findings, we investigate whether the preservation of plasticity is specific to the general $SU(4)$ ansatz employed in the main text or if it extends to standard, resource-constrained architectures common on near-term hardware.

We evaluate a hardware-efficient ansatz (HEA) on the Permuted MNIST benchmark. Unlike the main text model which uses arbitrary two-qubit unitaries, this ansatz enforces a structured topology with restricted gate sets, reducing the parameter density per layer.

The model operates on $N=10$ qubits with a circuit depth of $D=16$. Each layer consists of two sub-layers:
\begin{itemize}
    \item \textbf{Entanglement Layer:} Parameterized controlled-rotation-$X$ gates ($CR_X(\theta)$) are applied in a ring connectivity. For qubits indexed $i \in \{0, \dots, N-1\}$, gates are applied between pairs $(i, (i+1) \pmod N)$.
    \item \textbf{Rotation Layer:} A sequence of local rotations is applied to every qubit individually: $R_y(\theta_1) R_z(\theta_2) R_y(\theta_3)$.
\end{itemize}
The similar structures are widely used in quantum machine learning due to its compatibility with near-term hardware qubit connectivity. The training protocol (learning rate, optimizer, task sequence) remains identical to the Permuted MNIST experiment described in the main text.

The results over a sequence of 500 tasks are visualized in Fig.~\ref{fig:hea_plasticity}. The HEA-based QNN maintains a stable test accuracy, with no observable downward trend.
This result confirms that the loss of plasticity is not solved merely by the high expressivity of the $SU(4)$ gate, but by the fundamental properties of the quantum optimization manifold. Even with a restricted gate set, the parameters remain confined to the compact torus, preventing the weight saturation observed in classical models. It also demonstrates that the plasticity advantage persists in architectures designed for physical hardware implementation, suggesting that current-generation quantum devices could potentially realize these benefits.

\begin{figure}[h]
\centering
\includegraphics[width=0.5\textwidth]{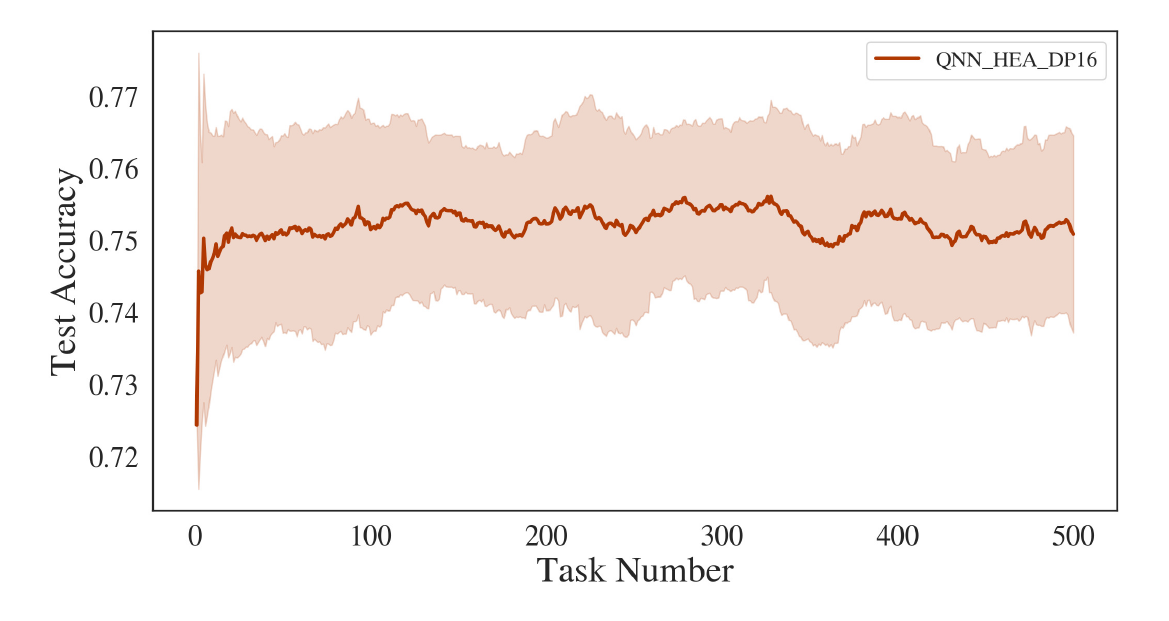}
\caption{\textbf{Generality of plasticity preservation across quantum architectures.} Test accuracy of a hardware-efficient ansatz with 16 layers over a 500-task sequence of permuted MNIST. The architecture utilizes a ring topology of parameterized $CR_x$ gates for entanglement, followed by local
$R_y$-$R_z$-$R_y$ sequence of rotations. Despite the restricted gate set compared to the $SU(4)$ model, the HEA maintains consistent learnability, confirming that the preservation of plasticity is a generic structural property of variational quantum circuits arising from their compact parameter space. The curve shows the moving
average over 40-task windows, with shaded areas representing
the standard deviation of the underlying data within each
window.}
\label{fig:hea_plasticity}
\end{figure}

\section{Control Experiment with Periodic Activation Functions}

To rigorously test whether the preservation of plasticity in QNNs arises merely from the bounded nature of the quantum state output ($|\langle Z \rangle| \le 1$) or fundamentally from the compactness of the parameter manifold, we perform a control experiment using a classical MLP with a periodic activation function.

We replace the standard ReLU activation with the sine function, $\phi(x) = \sin(x)$ \cite{Sitzmann2020}. Like the quantum expectation value, the output of a sine neuron is strictly bounded in $[-1, 1]$. If boundedness of the representation is the sole factor for plasticity, this model should also exhibit plasticity.

The experiment is conducted on the Split CIFAR-100 benchmark (3,000 tasks) using the same architecture as the standard MLP, but with the modifications that all hidden neurons use $\sin(x)$ activation. The epoch for training on each task is $20$.

The results are presented in Fig.~\ref{fig:sin_mlp}. We observe an evident decoupling between training and test performance and also a decay in test accuracy:
\begin{itemize}
    \item \textbf{Training Accuracy (Dashed Lines):} The model maintains near-perfect training accuracy ($100\%$) throughout the sequence. This indicates that the model retains high capacity.
    \item \textbf{Test Accuracy (Solid Lines):} Despite perfect training performance, the test accuracy degrades significantly over time, dropping from $\sim$75\% to $\sim$55\%.
\end{itemize}

 This behavior confirms the theoretical derivation before. Although the activation is bounded, the weights $\mathbf{w}$ reside in an unbounded Euclidean space. As the magnitude $||\mathbf{w}||$ increases to fit more tasks, the effective frequency of the sine neurons ($\sin(\mathbf{w}^T \mathbf{x})$) increases. This transforms the optimization landscape into a high-frequency, extremely rugged terrain. While the optimizer can find a narrow minimum for the current training batch (hence 100\% train accuracy), the resulting function generalizes poorly to the test set and becomes brittle to the distributional shifts inherent in continual learning.

Therefore, bounding the activation function is insufficient to preserve plasticity from classical side. The indefinite plasticity observed in QNNs is therefore due to the compactness of the parameter manifold, which prevents the unbounded growth of complexity that drives the generalization collapse in periodic classical networks.

\begin{figure}[h]
\centering
\includegraphics[width=0.5\textwidth]{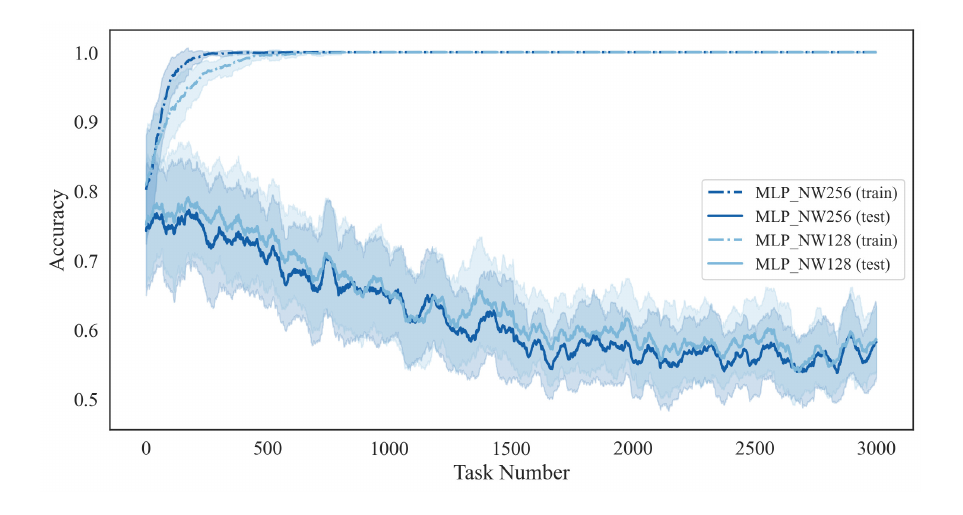} 
\caption{\textbf{Plasticity loss in classical MLPs with periodic activation on Split CIFAR-100.} The plot shows the evolution of training accuracy (dashed lines) and test accuracy (solid lines) over 3,000 binary classification tasks for network widths of 128 (light blue) and 256 (dark blue). Sin-MLPs maintain perfect training accuracy, while the test accuracy degrades continuously, exhibiting a wide generalization gap. This demonstrates that bounding the output representation in classical models is not sufficient to keep plasticity. All curves show the moving average over 50-task windows, with shaded areas representing the standard deviation of the underlying data within each window.}
\label{fig:sin_mlp}
\end{figure}

\section{Performance Results with Standard Deviation on Smooth Window}

In Fig. \ref{fig:cifar}\textbf{b} and Fig. \ref{fig:qdata}\textbf{b}, we present the moving average of test accuracy to highlight the long-term plasticity trends. To provide a complete picture of the learning stability and task-to-task variability, we present here the same results with full standard deviation bands included.

Fig.~\ref{fig:cifar_variance} illustrates the results for the Split CIFAR-100 benchmark. The shaded regions represent the standard deviation of the test accuracy calculated within each 150-task moving window. The large band demonstrates the great variability of difficulty for image classification from different classes. Fig.~\ref{fig:cifar_variance}\textbf{b} directly shows the slope estimation of QNN accuracy line. The estimated slope is negligible ($\approx 0.01\%$ drop per 100 tasks). Crucially, the $p$-value associated with this slope estimation is \textbf{$p \gg 0.05$}. This statistical evidence confirms that the QNN performance is indistinguishable from a stationary process. The non-significant slope implies that the model is not slowly dying, but rather fluctuating around a stable equilibrium of learnability, effectively preserving plasticity indefinitely.

Fig.~\ref{fig:qdata_variance} presents the corresponding analysis for the quantum data continual learning task. 

\begin{figure*}[h]
\centering
\includegraphics[width=0.8\textwidth]{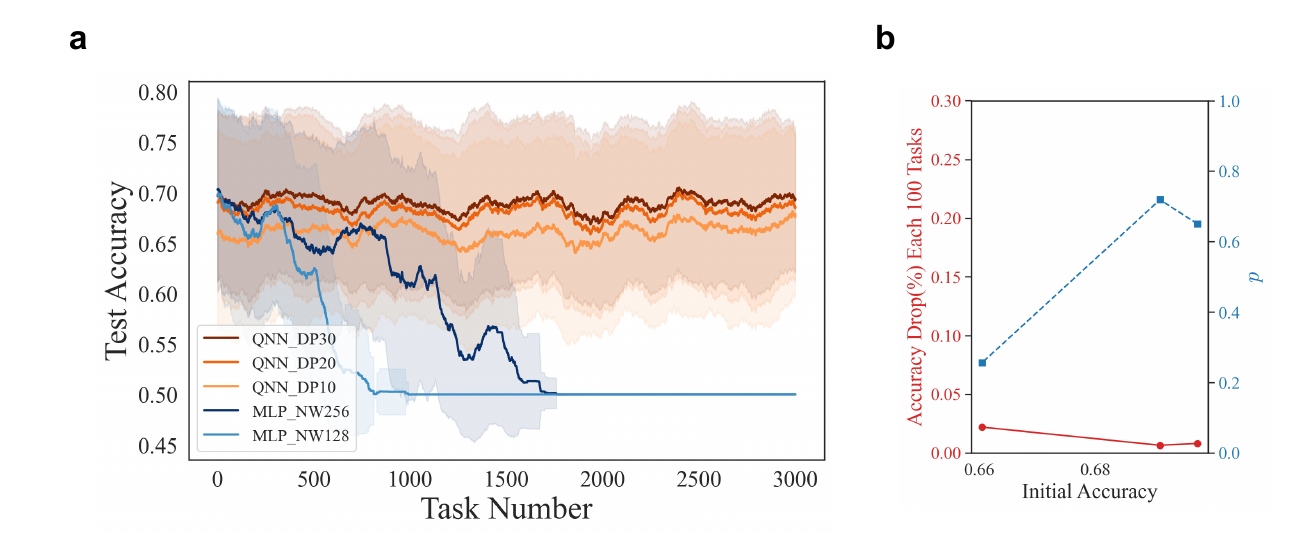}
\caption{\textbf{Plasticity dynamics with performance variability on Split CIFAR-100.} \textbf{a}, Test accuracy as a function of task number for QNNs and MLPs, identical to Fig. 2b but visualizing the standard deviation (shaded regions) within each 150-task window. \textbf{b}, The rate of plasticity loss (accuracy drop per 100 tasks) versus initial accuracy for QNN. Blue circles represent the corresponding $p$ value for the linear estimation whose null hypothesis is zero slope.}
\label{fig:cifar_variance}
\end{figure*}

\begin{figure*}[h]
\centering
\includegraphics[width=0.8\textwidth]{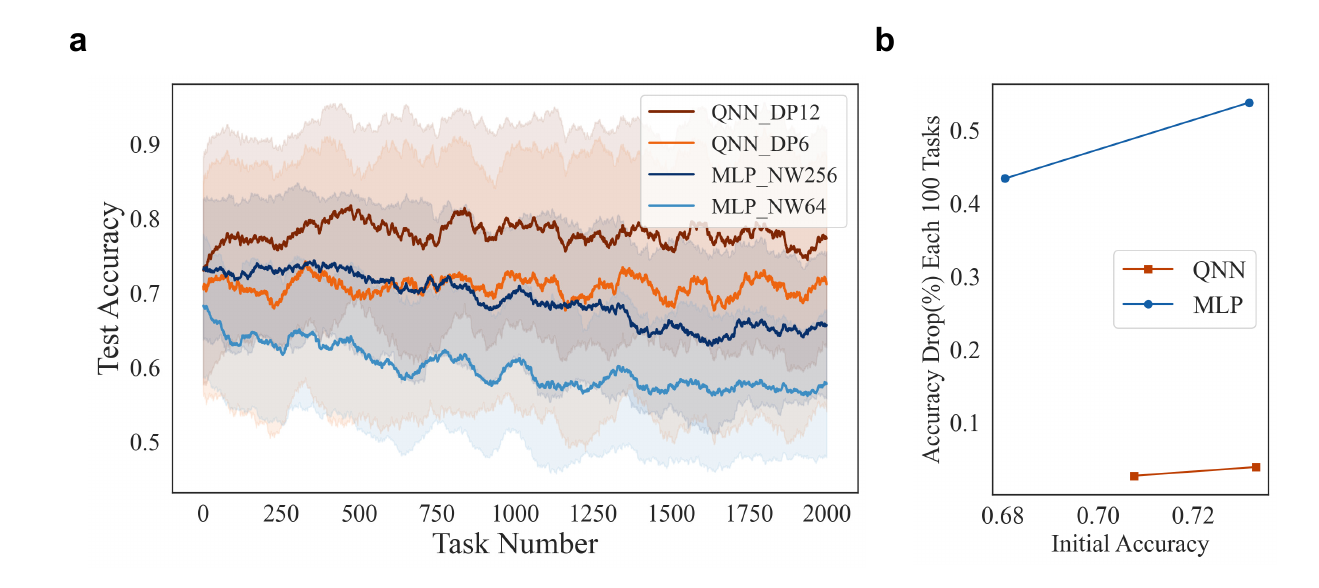}
\caption{\textbf{Plasticity dynamics with performance variability on Quantum Data.} \textbf{a}, Test accuracy as a function of task number, identical to Fig. 4b but including the standard deviation (shaded regions) over 100-task moving windows. \textbf{b}, Average accuracy drop per 100 tasks plotted against initial accuracy. The large value for MLPs confirms a consistent loss of plasticity, while QNNs maintain a near-zero drop rate, confirming the plasticity preservation.}
\label{fig:qdata_variance}
\end{figure*}

\section{Hyperparameters in architectures and setups}
To ensure the reproducibility of our results and to demonstrate the universality of the plasticity preservation mechanism, we provide a detailed breakdown of the hyperparameters, network architectures, and initialization strategies employed across all four experimental domains: Permuted MNIST, Split CIFAR-100, Reinforcement Learning (Ant-v4), and Quantum Data Classification.

Crucially, these experiments span a wide range of configurations:
\begin{itemize}
    \item \textbf{Model Complexity:} From compact QNNs ($\sim$200 parameters) to large classical MLPs ($>$500,000 parameters).
    \item \textbf{Optimization Regimes:} From standard supervised learning with Adam to complex PPO updates.
    \item \textbf{Initialization Schemes:} From standard Glorot/Orthogonal initialization to simple uniform distributions.
\end{itemize}

Despite this heterogeneity in hyperparameters across tasks, the qualitative difference in learning dynamics, i.e., classical collapse versus quantum plasticity, remains consistent across all settings. This suggests that the observed advantage is not an artifact of specific tuning but arises from the fundamental structural differences.

The initial state of a network can significantly influence its training trajectory. As detailed in Table \ref{tab:initialization}, we employ industry-standard initialization schemes for the classical baselines to ensure they started from optimal conditions.
Even with these stability-enhancing initializations, classical models fail to preserve plasticity over long task sequences. Conversely, QNNs utilize simple uniform or normal distributions for their rotation angles. 

\begin{table*}[h]
\centering
\caption{\textbf{Weight initialization strategies across experiments.} }
\label{tab:initialization}
\renewcommand{\arraystretch}{1.3}
\begin{tabular}{lll}
\hline
\textbf{Experiment} & \textbf{Model} & \textbf{Weight Initialization Scheme} \\
\hline
\textbf{Permuted MNIST} & Classical MLP & Glorot Uniform  \\
 & Deep QNN & Uniform Distribution $\mathcal{U}[0, 2\pi]$ \\
\hline
\textbf{Split CIFAR-100} & Classical MLP & Glorot Uniform  \\
 & Deep QNN & Uniform Distribution $\mathcal{U}[0, 2\pi]$ \\
\hline
\textbf{RL Ant-v4} & Classical Agent & Orthogonal Initialization (Gain=$\sqrt{2}$) \\
 & Quantum Agent & Standard Normal Distribution $\mathcal{N}(0, 1)$ \\
\hline
\textbf{Quantum Data} & Classical MLP & Glorot Uniform  \\
 & Deep QNN & Uniform Distribution $\mathcal{U}[-0.1, 0.1]$ \\
\hline
\end{tabular}
\end{table*}

The specific hyperparameters for the RL agents are listed in Table \ref{tab:ppo_hyperparams}. To make a fair comparison, all the RL learning hyperparameters listed here are the same for classical and quantum agents.

\begin{table*}[h]
\centering
\caption{\textbf{Hyperparameter settings for the PPO agents in the Ant-v4 environment in RL task.}}
\label{tab:ppo_hyperparams}
\begin{tabular}{lcc}
\hline
\textbf{Hyperparameter} & \textbf{Classical and Quantum Agent}  \\
\hline
Horizon (Steps per rollout) & 2048  \\
Mini-batch Size & 128 \\
Number of Epochs (per update) & 10 \\
Discount Factor  & 0.99\\
GAE Parameter & 0.95 \\
Clip Range ($\epsilon$) & 0.2\\
Value Loss Coefficient ($c_1$) & 0.5  \\
Entropy Coefficient ($c_2$) & 0.0 \\
Max Gradient Norm & 0.5\\
Parallel Environments & 16  \\
Total Timesteps & 30,000,000 \\
\hline
\end{tabular}
\end{table*}

Table \ref{tab:param_counts} provides a comparative summary of the trainable parameter counts for all models. We note that QNNs maintain plasticity with orders of magnitude fewer parameters, highlighting the efficiency of QML.

To track the underlying capacity dynamics, we perform a detailed model Fisher Information Matrix analysis on the extended 3000-task Split CIFAR-100 sequence. The corresponding tracking hyperparameters, including analysis frequency and the effective rank threshold, are summarized in Table \ref{tab:analysis_hyperparams}.

\begin{table*}[h]
\centering
\caption{\textbf{Comparison of model complexity in terms of total trainable parameters.} Parameter counts for Classical MLPs are derived from the specific architectures (hidden widths) described in the Methods. For RL, the count includes both the policy and value networks. For QNNs, counts are based on the number of $SU(4)$ gates (15 parameters per gate) plus any classical trainable readout heads. QNNs achieve plasticity with orders of magnitude fewer parameters.}
\label{tab:param_counts}
\renewcommand{\arraystretch}{1.5}
\begin{tabular}{lllr}
\hline
\textbf{Experiment} & \textbf{Model} & \textbf{Configuration} & \textbf{Total Parameters} \\
\hline
\textbf{Permuted MNIST} & Classical MLP & Width 16 (1 Hidden Layer) & $\approx$ 13,000 \\
 & Deep QNN & Depth 16 (10 Qubits) & $\approx$\textbf{2,200} \\
\hline
\textbf{Split CIFAR-100} & Classical MLP & Width 256 (2 Hidden Layers) & $\approx$ 330,000 \\
 & Deep QNN & Depth 30 (10 Qubits) & $\approx$\textbf{4,100} \\
\hline
\textbf{RL Ant-v4} & Classical Agent & Policy + Value Networks & $\approx$ 370,000 \\
 & Hybrid Quantum Agent & QNN (12 Layers) + Linear Heads & $\approx$ \textbf{6,200} \\
\hline
\textbf{Quantum Data} & Classical MLP & Width 256 (2 Hidden Layers) & $\approx$ 590,000 \\
 & Deep QNN & Depth 12 (10 Qubits) & $\approx$ \textbf{1600} \\
\hline
\end{tabular}
\end{table*}

\begin{table*}[h]
\centering
\caption{\textbf{Hyperparameter settings for the Information Plasticity Analysis (Model FIM analysis).} These settings apply to the extended 3,000-task investigation on Split CIFAR-100 benchmarking plasticity preservation mechanism.}
\label{tab:analysis_hyperparams}
\begin{tabular}{lc}
\hline
\textbf{Hyperparameter} & \textbf{Value} \\
\hline
Analysis Frequency & Every 5 tasks \\
Probe Sample Size ($M$) & 64 (Task 0 Training Data) \\
Effective Rank Threshold & $\lambda_i > 10^{-4} \cdot \text{Tr}(F)$ \\
Total CL Tasks & 3,000 \\
Optimizer & Adam ($10^{-3}$) \\
Batch Size & 128 \\
\hline
\end{tabular}
\end{table*}

\section{Supplementary Benchmarks and Robustness Analysis}
To quantify the robustness of the quantum plasticity advantage and isolate its structural origins, we conduct targeted comparative benchmarks. 

\textbf{L2 regularization.} We examine whether the loss of plasticity in classical networks is simply a consequence of unregularized weight growth. To test this, we evaluate a classical MLP with varying degrees of explicit weight decay on the Split CIFAR-100 task sequence. As shown in Fig.~\ref{fig:R1}, $L_2$ regularization fails to preserve plasticity; the models still eventually collapse to the random guessing level ($0.50$).The effectiveness of this approach is highly sensitive to the choice of the weight decay hyperparameter. While small weight decays ($10^{-4}$ and $10^{-3}$) offer negligible help in maintaining accuracy over hundreds of tasks, a larger decay ($10^{-2}$) is overly restrictive, causing the network to fail to learn meaningful representations from the very beginning. The observations are consistent with Ref. \cite{Dohare2024}. This sensitivity suggests that $L_2$ regularization cannot fundamentally prevent the accumulation of dormant units or the decline of the representation's effective rank. These results highlight that extrinsic constraints in Euclidean space do not resolve the fundamental geometric saturation problem of the classical landscape.

\begin{figure}[h]
    \centering
    \includegraphics[width=0.45\textwidth]{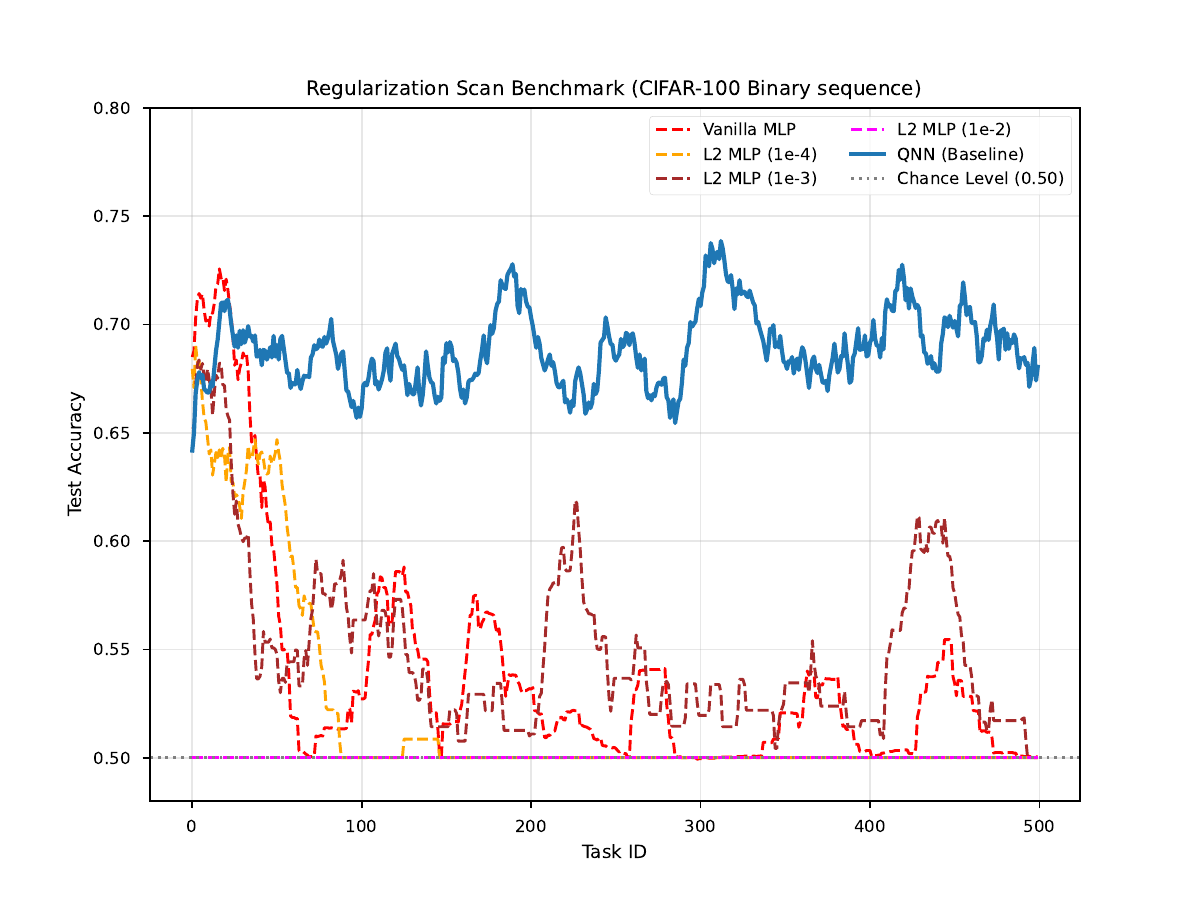}
    \caption{\textbf{Plasticity preservation under explicit regularization.} Evolution of test accuracy on the CIFAR-100 Binary sequence benchmark across 500 tasks (Adam optimizer with learning rate $0.005$, 5 epochs for one task). The performance of the baseline QNN is compared against a Vanilla MLP and classical MLPs with varying degrees of L2 regularization. While the QNN sustains plasticity and maintains accuracy over time, explicit regularization fails to prevent catastrophic forgetting in the MLPs. Notably, small L2 weight decays ($10^{-4}$ and $10^{-3}$) offer little help in preserving plasticity, with accuracy eventually degrading to the chance level (0.50). Conversely, a large L2 weight decay ($10^{-2}$) is overly restrictive, causing the network to directly fail learning from the very beginning and remain at the chance level. }
    \label{fig:R1}
\end{figure}

\textbf{NISQ noise robustness.} We verify the mechanism's robustness on hardware-realistic noisy intermediate-scale quantum devices by simulating a depolarizing noise model ($p_x=p_y=p_z=5 \times 10^{-4}$). As visualized in Fig.~\ref{fig:R3}, the QNN maintains a stable learning performance despite the systematic reduction in absolute baseline accuracy. This confirms that the geometric preservation of plasticity is robust to small quantum errors.

\begin{figure}[h]
    \centering
    \includegraphics[width=0.45\textwidth]{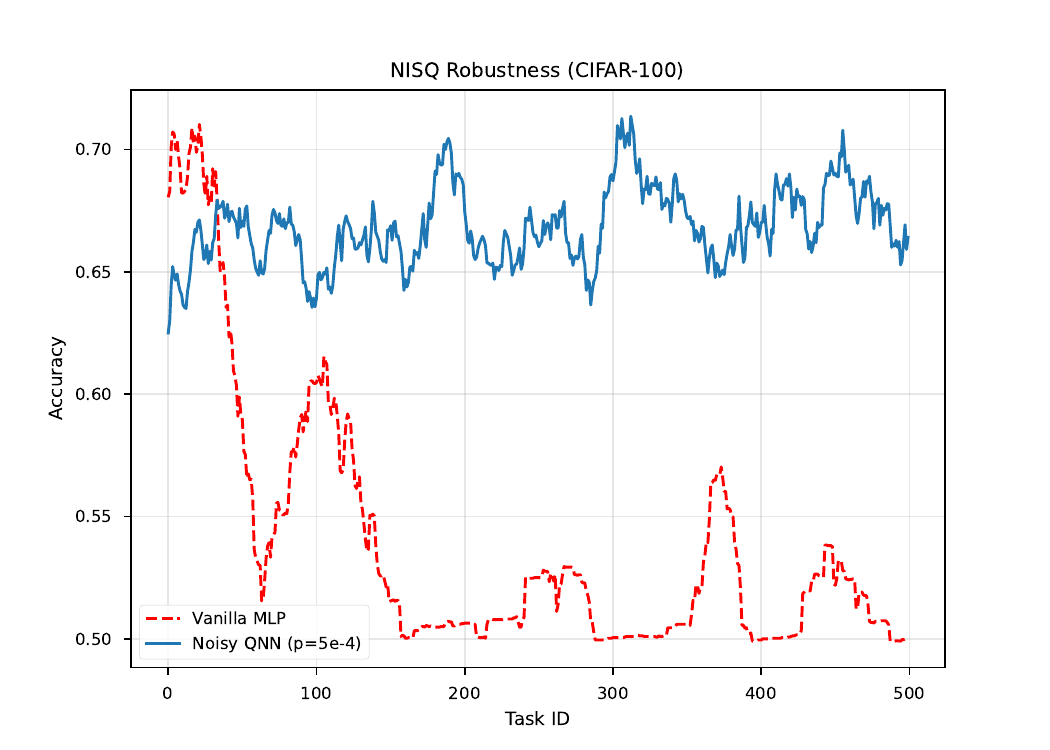}
    \caption{\textbf{Robustness to hardware noise in the NISQ era.} Continual learning performance of the QNN under a depolarizing noise model (noise ratio $p_x=p_y=p_z=5 \times 10^{-4}$, Adam optimizer with learning rate $0.005$, 5 epochs for one task) on the Split CIFAR-100 benchmark. Despite the reduction in baseline performance, the QNN maintains its learnability over time, contrasting with the dynamic failure of classical architectures.}
    \label{fig:R3}
\end{figure}

\begin{figure*}[t]
	\centering
	\includegraphics[width=0.8\textwidth]{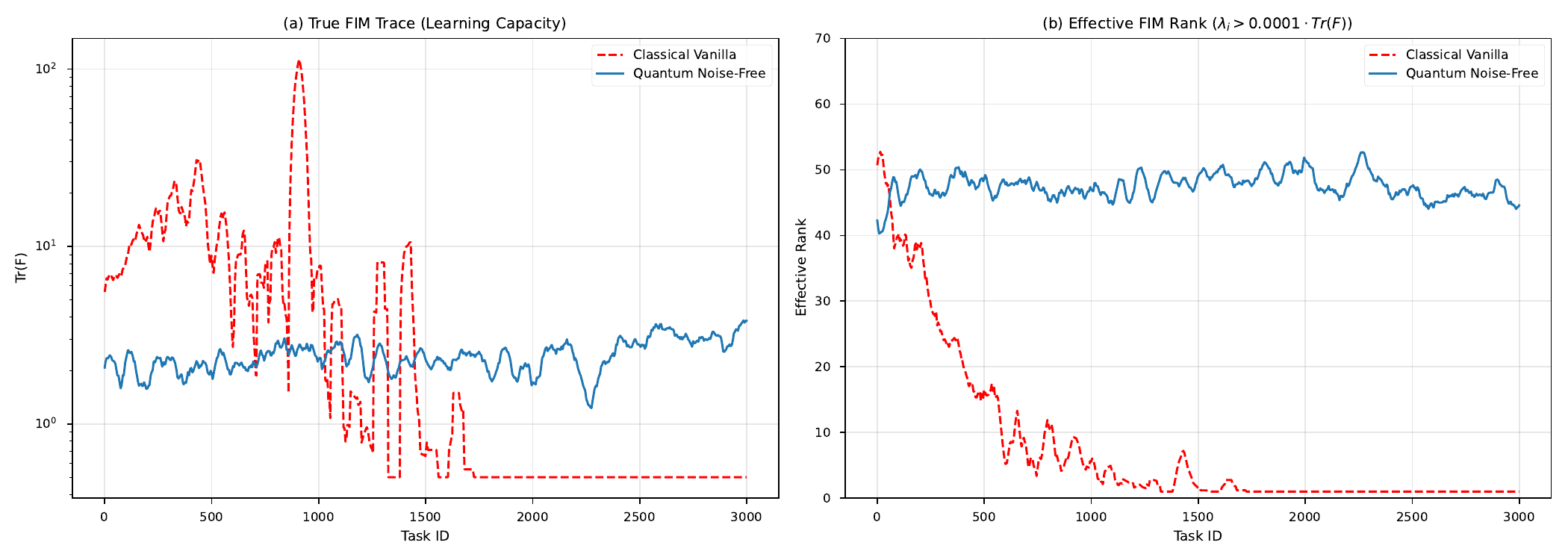}
	\caption{\textbf{Information-theoretic assessment of learning capacity.} (a) Evolution of the trace of the model FIM over 3,000 sequential tasks on the Split CIFAR-100 benchmark. (b) Monitoring of the effective FIM rank ($\lambda_i > 10^{-4} \text{Tr}(F)$). The functional dimensionality of the classical MLP (red dashed) collapses to unity, while the QNN (blue) preserves a high-rank landscape throughout the sequential learning process, providing a structural explanation for the persistent plasticity advantage.}
	\label{fig:appendix_fim}
\end{figure*}

\section{Model Fisher Information and Rank Analysis}
To provide a more fundamental information-theoretic perspective on the loss of plasticity, we conduct an additional analysis using the \textit{Model Fisher Information Matrix}, as opposed to the empirical FIM presented in the main text. Mathematically, while the empirical FIM is computed using the specific targets $y_i$ present in the batch (representing the local curvature of the training loss surface), the model FIM is defined as the expectation over the model's own predictive distribution:
\begin{equation}
F = \mathbb{E}_{\mathbf{x} \sim p(\mathbf{x})} \mathbb{E}_{y \sim p(y|\mathbf{x}, \theta)} [\nabla_\theta \log p(y|\mathbf{x}, \theta) \nabla_\theta \log p(y|\mathbf{x}, \theta)^T].
\end{equation}
Physically, the model FIM represents the Riemannian metric of the model's manifold, measuring the sensitivity of the entire output distribution to parameter perturbations. 

We observe that the model FIM and the empirical FIM exhibit qualitatively identical behavior in our lifelong learning experiments. This alignment stems from the fact that our models are consistently trained to a high-accuracy regime. When the model predictive distribution $p(y|\mathbf{x}, \theta)$ is well-optimized and highly confident (matching the true data distribution), the expectation over $y$ in the model FIM calculation becomes dominated by the terms corresponding to the ground-truth labels. In this high-confidence regime, the information-theoretic curvature of the likelihood and the curvature of the actual loss surface converge, rendering the empirical FIM a faithful and computationally efficient proxy for the fundamental learning capacity.

In our extended 3,000-task investigation on Split CIFAR-100, we monitor both the trace of the model FIM and its \textit{effective rank}, defined as the number of eigenvalues $\lambda_i$ that satisfy $\lambda_i > 10^{-4} \cdot \text{Tr}(F)$. The analysis is conducted using a fixed probe set of $M=64$ samples from the initial task, with the FIM being computed every 5 tasks. As shown in Fig.~\ref{fig:appendix_fim}, the model FIM trace exhibits a qualitatively identical decay pattern to the empirical FIM in classical networks. More significantly, the effective rank of the classical MLP undergoes a catastrophic collapse, dropping from its initial value of 64 to exactly 1 by Task 1,500. This indicates that the classical model's functional parameter space effectively implodes, leaving only a single effective direction for learning. In contrast, the QNN maintains a high and stable effective rank (averaging $\sim$ 48) across the entire sequence, confirming that the unitary manifold preserves a high-dimensional landscape for continuous adaptation. These results support that the pattern of the trace and the effective rank of the FIM is qualiatively similar and both are reliable proxies for the functional plasticity of the learner.

\clearpage
\newpage
\let\oldaddcontentsline\addcontentsline
\renewcommand{\addcontentsline}[3]{}

\let\addcontentsline\oldaddcontentsline

\end{document}